\documentclass[12pt,a4paper]{article}%
\usepackage{hyperref}
\usepackage{amssymb}
\usepackage{latexsym}
\usepackage{epsfig}
\usepackage{dcolumn}
\usepackage{amsmath}
\usepackage{amsfonts}
\usepackage{bm}
\usepackage{color}
\usepackage{graphicx}%

\newtheorem{theorem}{Theorem}
\newtheorem{proposition}[theorem]{Proposition}
\newtheorem{corollary}[theorem]{Corollary}

\newtheorem{lemma}[theorem]{Lemma}
\newtheorem{remark}[theorem]{Remark}
\newenvironment{proof}[1][Proof]{\textbf{#1.} }{\ \rule{0.5em}{0.5em}}
\begin{document}

\title{Form factors of the $O(6)$ Gross Neveu-model}
\author{Hrachya M. Babujian\thanks{Address: Yerevan Physics Institute, Alikhanian
Brothers 2, Yerevan, 375036 Armenia} \thanks{E-mail: babujian@yerphi.am} ,
Angela Foerster\thanks{Address: Instituto de F\'{\i}sica da UFRGS, Av. Bento
Gon\c{c}alves 9500, Porto Alegre, RS - Brazil} \thanks{E-mail:
angela@if.ufrgs.br} , and Michael Karowski\thanks{E-mail:
karowski@physik.fu-berlin.de}\\Institut f\"{u}r Theoretische Physik, Freie Universit\"{a}t Berlin,\\Arnimallee 14, 14195 Berlin, Germany }
\date{\today\\
[1cm]{\textit{In memory of Petr Kulish}}}
\maketitle

\begin{abstract}
The isomorphism $SU(4)\simeq O(6)$ is used to construct the form factors of
the $O(6)$ Gross-Neveu model as bound state form factors of the $SU(4)$ chiral
Gross-Neveu model. This technique is generalized and is then applied to use
the $O(6)$ as the starting point of the nesting procedure to obtain the $O(N)$
form factors for general even $N$.

\end{abstract}
\tableofcontents

\section{Introduction}

In the last decades integrable quantum field theories in 1+1 dimensions have
been investigated very intensively. One of the pioneers was Petr Kulish: An
infinite set of conservation laws for the sine-Gordon and the massive Thirring
model was derived by Kulish and Nissimov in \cite{KN} (see also \cite{AKNP}).
In \cite{Ku}, P. Kulish has shown that these conservation laws imply the
factorization of the S-matrix. He also made a seminal contribution in the
algebraic formulation of the nested Bethe ansatz: in \cite{KR} Kulish and
Reshetikhin constructed the nested version of the algebraic Bethe ansatz for a
$GL(N)$ invariant model. The \textquotedblleft off-shell\textquotedblright
version of this nested algebraic Bethe ansatz was later developed in
\cite{BKZ2} to solve matrix difference equations. This technique was applied
in \cite{BFK1,BFK2,BFK3} to construct form factors for the $SU(N)$ chiral
Gross-Neveu model.

In a previous paper \cite{BFK6} we constructed the $O(N)$ nested Bethe ansatz,
which needs deeper investigations. We introduced an intertwiner, which
connects two different S-matrices in the nesting procedure $S(\theta,N)$ and
$S(\theta,N-2)$. Then we applied this technique in \cite{BFK7} and \cite{BFK8}
to the $O(N)$ nonlinear $\sigma$-model and the $O(N)$ Gross-Neveu model with
even $N$, respectively. In the present article we will consider the form
factors of the $O(6)$ Gross-Neveu model which will be the starting model for
the nesting procedure for the $O(N)$ Gross-Neveu model. The $O(4)$
Gross-Neveu-Model will be considered in forthcoming papers.

Our results are related to the $N=4$ supersymmetric Yang-Mills (SYM) theory.
It is known that the $O(6)$ or $SU(4)$ Bethe ansatz structure is connected to
the $N=4$ SYM theory, which, in turn, is equivalent by the AdS/CFT conjecture
to the super-string theory on the product space $AdS_{5}\times S_{5}$. This
equivalence means that there is a one-to-one correspondence between all
aspects of the theories including the global symmetry observables and the
field content with correlation functions. In the $N=4$ SYM theory there is an
automorphism symmetry group of the supersymmetry algebra known as R-symmetry,
which causes the supercharges to change by a phase rotation. Thus for the
$N=4$ SYM theory the R-symmetry group is $SU(4)\simeq O(6)$. This group is
part of the full group of symmetry of the theory known as superconformal group
and is given by $S(2,2\,|\,4)$ which also includes the conformal subgroup
$SO(2,4)$ and Poincare supersymmetry \cite{Mal,KZar}. Therefore all integrable
structures associated with $SU(4)\simeq O(6)$ group are interesting tools for
this big AdS/CFT correspondence conjecture.

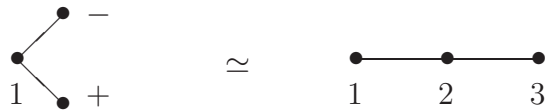
\begin{figure}[th]%
\[%
\begin{array}
[c]{c}%
\unitlength=6mm\begin{picture}(2,2)(8,0) \put(8,1){\makebox(0,0){$\bullet$}} \put(8,1){\line(1,1){1}} \put(8,1){\line(1,-1){1}} \multiput(9,0)(0,2){2}{\makebox(0,0){$\bullet$}} \put(7.8,0){1} \put(9.5,0){$+$} \put(9.5,1.8){$-$} \end{picture}
\end{array}
\qquad~~~\simeq~~\qquad%
\begin{array}
[c]{c}%
\unitlength=6mm\begin{picture}(4,2)(0,0) \put(0,1){\line(1,0){4}} \multiput(0,1)(2,0){3}{\makebox(0,0){$\bullet$}} \put(-0.2,0){1} \put(1.8,0){2} \put(3.8,0){3} \end{picture}
\end{array}
~
\]
\caption{The isomorphism $0(6)\simeq SU(4)$ in terms of the Dynkin diagrams}%
\label{f1}%
\end{figure}In \cite{KT1} was shown that the isomorphism $O(6)\simeq SU(4)$,
cf. Fig. \ref{f1}, leads to an identity between the $O(6)$ Gross-Neveu model
and the $SU(4)$ chiral Gross-Neveu model. The four right-handed (left-handed)
$O(6)$ kinks correspond to the four fundamental $SU(4)$ particles
(antiparticles). The six fundamental $O(6)$ particles correspond to the six
$SU(4)$ bound states. In \cite{KT1} the isomorphism was shown for the
S-matrices. In this article we demonstrate the isomorphism for the form factors.

In \cite{BFK7} and \cite{BFK8} we constructed form factors for the
$O(N)~\sigma$-model and the $O(N)$ Gross-Neveu model (for $N$ even),
respectively. For these constructions we used the nested Bethe ansatz, which
means that for the level $N$ one needs the results from level $N-2$, etc. In
\cite{BFK7} we used the isomorphism $O(4)\simeq SU(2)\times SU(2)$ as the
starting point of the nesting procedure for the $O(N)~\sigma$-model. The
$SU(N)$ form factors were constructed in \cite{BFK0,BFK1,BFK2,BFK3,BFK4}. The
results of the present article, which relay on the isomorphism $O(6)\simeq
SU(4)$ may serve as the starting point of the nesting procedure for the $O(N)$
Gross-Neveu model.

The article is organized as follows. In Section \ref{s2} we recall some
results on the needed S-matrices, in particular the bound state procedure. In
Section \ref{s3} we recall results on the $SU(4)$ and $O(6)$ form factors. We
show that the form factors for $O(6)$ vector particles are to be identified
with bound state $SU(4)$ form factors. In Section \ref{s4} we apply these
results to some examples. In Section \ref{s5} we generalize the results to the
so called `modified form factors'. We prove that they can be used to start the
nested `off-shell' Bethe ansatz to solve the $O(N)$ form factor equations. The
appendix provides the more complicated proofs of the results we have obtained
and further explicit calculations.

\section{S-matrix}

\label{s2}

\subsection{The $SU(4)$ S-matrix}

The S-matrix of the $SU(4)$ chiral Gross-Neveu model for the scattering of two
fundamental particles (transforming as the $SU(4)$ vector representation) is
given by \cite{BKKW,BW,KKS,KT1,BFK1}%
\begin{equation}
S^{SU(4)}(x)=b^{SU(4)}(\theta)\,\mathbf{1}+c^{SU(4)}(\theta)\,\mathbf{P}
\label{S4}%
\end{equation}
or in terms of the components
\[
\left(  S^{SU(4)}\right)  _{AB}^{DC}(x)=b^{SU(4)}(\theta)\,\delta_{A}%
^{C}\delta_{B}^{D}+c^{SU(4)}(\theta)\,\delta_{A}^{D}\delta_{B}^{C}\,=%
\begin{array}
[c]{c}%
\unitlength3mm\begin{picture}(6,6) \put(1,1){\line(1,1){4}} \put(5,1){\line(-1,1){4}} \put(.5,0){$A$} \put(5,0){$B$} \put(5,5.4){$C$} \put(.5,5.4){$D$} \put(1.8,.7){$\theta_1$} \put(3.5,.7){$\theta_2$} \end{picture}
\end{array}
\]
with the rapidity difference of the particles $\theta=\theta_{12}=\theta
_{1}-\theta_{2}$. The two S-matrix eigenvalues are $S_{\pm}^{SU(4)}%
=b^{SU(4)}\pm c^{SU(4)}$ with%
\begin{equation}
\left(  S_{+}^{SU(4)},S_{-}^{SU(4)}\right)  =\left(  \frac{\theta-\frac{1}%
{2}i\pi}{\theta+\frac{1}{2}i\pi},1\right)  S_{-}^{SU(4)}\,. \label{Ev4}%
\end{equation}
Unitarity writes as%
\[
S_{+,-}^{SU(4)}(-\theta)S_{+,-}^{SU(4)}(\theta)=1\,.
\]
The highest weight amplitude
\begin{equation}
a^{SU(4)}(\theta)=S_{+}^{SU(4)}(\theta)=-\frac{\Gamma\left(  1-\frac{1}%
{2}\frac{\theta}{i\pi}\right)  \Gamma\left(  \frac{3}{4}+\frac{1}{2}%
\frac{\theta}{i\pi}\right)  }{\Gamma\left(  1+\frac{1}{2}\frac{\theta}{i\pi
}\right)  \Gamma\left(  \frac{3}{4}-\frac{1}{2}\frac{\theta}{i\pi}\right)  }
\label{a4}%
\end{equation}
is essential for the Bethe ansatz which will be used to construct the form
factors. In order to simplify the formulae we extract the factor
$a^{SU(4)}(\theta)$ from the S-matrix and define%
\begin{equation}
\tilde{S}^{SU(4)}(\theta)=S^{SU(4)}(\theta)/a^{SU(4)}(\theta)=\tilde
{b}^{SU(4)}(\theta)\mathbf{1}+\tilde{c}^{SU(4)}(\theta)\mathbf{P} \label{St4}%
\end{equation}
with%
\begin{equation}
\tilde{b}^{SU(4)}(\theta)=\frac{\theta}{\theta-\frac{1}{2}i\pi}\,,~~\tilde
{c}^{SU(4)}(\theta)=\frac{-\frac{1}{2}i\pi}{\theta-\frac{1}{2}i\pi
}\,.\nonumber
\end{equation}

The S-matrix eigenvalue $S_{-}^{SU(4)}(\theta)$ has a pole at $\theta=\frac
{1}{2}i\pi$ which means that there exist a bound state of 2 fundamental
particles, which transforms as an $SU(4)$ anti-symmetric tensor. This have to
identified with a fundamental particle of the $O(6)$ model (see below). The
bound states of 3 fundamental particles $(ABC)$ (with $1\leq A<B<C\leq4$) is
to be identified with an anti-particle of a fundamental particle
$D:(ABC)=\bar{D}$ \cite{KuS,KKS,BFK1}. The charge conjugation matrix is
\begin{equation}
\mathbf{C}_{(ABC)D}=\epsilon_{ABCD}\label{C4}%
\end{equation}
where $\epsilon_{ABCD}$ is total anti-symmetric and $\epsilon_{1234}=1$.

\subsection{The $O(6)$ S-matrix}

The $O(6)$ Gross-Neveu S-matrix for the scattering of two fundamental
particles (transforming as the $O(6)$ vector representation) writes as
\cite{ZZ4}%
\begin{equation}
S^{O(6)}(\theta)=b^{O(6)}(\theta)\mathbf{1}+c^{O(6)}(\theta)\mathbf{P}%
+d^{O(6)}(\theta)\mathbf{K\,,} \label{S}%
\end{equation}
or in terms of the components as%
\[
\left(  S^{O(6)}\right)  _{\alpha\beta}^{\delta\gamma}(\theta)=b^{O(6)}%
(\theta)\delta_{\alpha}^{\gamma}\delta_{\beta}^{\delta}+c^{O(6)}(\theta
)\delta_{\alpha}^{\delta}\delta_{\beta}^{\gamma}+d^{O(6)}(\theta
)\mathbf{C}^{\delta\gamma}\mathbf{C}_{\alpha\beta}=%
\begin{array}
[c]{c}%
\unitlength3mm\begin{picture}(6,6) \put(1,1){\line(1,1){4}} \put(5,1){\line(-1,1){4}} \put(.5,0){$\alpha$} \put(5,0){$\beta$} \put(5,5.4){$\gamma$} \put(.5,5.4){$\delta$} \put(1.8,.7){$\theta_1$} \put(3.5,.7){$\theta_2$} \end{picture}
\end{array}
\]
with the \textquotedblleft charge conjugation matrices\textquotedblright\ \
\begin{equation}
\mathbf{C}_{\alpha\beta}=\delta_{\alpha\bar{\beta}}~\text{and }\mathbf{C}%
^{\alpha\beta}=\delta^{\alpha\bar{\beta}}\, \label{C}%
\end{equation}
in the complex basis (see \cite{BFK8}). The three S-matrix eigenvalues are
$S_{\pm}^{O(6)}=b^{O(6)}\pm c^{O(6)}$ and $S_{0}^{O(6)}=b^{O(6)}%
+c^{O(6)}+6d^{O(6)}$ with%
\begin{equation}
\left(  S_{0}^{O(6)},S_{+}^{O(6)},S_{-}^{O(6)}\right)  =\left(  \frac
{\theta+i\pi}{\theta-i\pi},\frac{\theta-\frac{1}{2}i\pi}{\theta+\frac{1}%
{2}i\pi},1\right)  S_{-}^{O(6)}\,. \label{SO6}%
\end{equation}
Unitarity reads as%
\[
S_{0,+,-}^{O(6)}(-\theta)S_{0,+,-}^{O(6)}(\theta)=1\,.
\]
The highest weight amplitude is \cite{ZZ4}%
\begin{equation}
a^{O(6)}(\theta)=S_{+}^{O(6)}(\theta)=\frac{\Gamma\left(  1-\frac{1}{2\pi
i}\theta\right)  \Gamma\left(  \frac{1}{2}+\frac{1}{2\pi i}\theta\right)
}{\Gamma\left(  1+\frac{1}{2\pi i}\theta\right)  \Gamma\left(  \frac{1}%
{2}-\frac{1}{2\pi i}\theta\right)  }\frac{\Gamma\left(  \frac{3}{4}+\frac
{1}{2\pi i}\theta\right)  \Gamma\left(  \frac{1}{4}-\frac{1}{2\pi i}%
\theta\right)  }{\Gamma\left(  \frac{3}{4}-\frac{1}{2\pi i}\theta\right)
\Gamma\left(  \frac{1}{4}+\frac{1}{2\pi i}\theta\right)  }\,.\nonumber
\end{equation}
For later convenience we introduce again
\begin{equation}
\tilde{S}^{O(6)}(\theta)=S^{O(6)}(\theta)/a^{O(6)}(\theta)=\tilde{b}%
^{O(6)}(\theta)\mathbf{1}+\tilde{c}^{O(6)}(\theta)\mathbf{P}+\tilde{d}%
^{O(6)}(\theta)\mathbf{K} \label{StO6}%
\end{equation}
with%
\begin{equation}
\tilde{b}^{O(6)}(\theta)=\frac{\theta}{\theta-\frac{1}{2}i\pi},~\tilde
{c}^{O(6)}(\theta)=\frac{\mathbf{-}\frac{1}{2}i\pi}{\theta-i\pi},~\tilde
{d}^{O(6)}(\theta)=\frac{\theta}{\theta-\frac{1}{2}i\pi}\frac{\mathbf{-}%
\frac{1}{2}i\pi}{i\pi-\theta}\,.\nonumber
\end{equation}

\begin{remark}
Note, that the amplitudes $\tilde{b}$ and $\tilde{c}$ are the same for $SU(4)$
and $O(6)$.
\end{remark}

\subsection{Bound state S-matrix}

The S-matrix eigenvalue $S_{-}^{SU(4)}(\theta)$ of (\ref{Ev4}) has a pole at
$\theta=\frac{1}{2}i\pi$ wich means that two fundamental particles $A$ and $B$
form an anti-symmetric tensor bound state $\left(  AB\right)  $. The S-matrix
for the scattering of these bound states with fundamental particles is given
by \cite{K1,KT1}%
\begin{align}
S_{(RS)C}^{C^{\prime}(R^{\prime}S^{\prime})}(\theta_{(12)3})\Gamma
_{AB}^{(RS)}  &  =\Gamma_{A^{\prime}B^{\prime}}^{(R^{\prime}S^{\prime}%
)}S_{AC^{\prime\prime}}^{C^{\prime}A^{\prime}}(\theta_{13})S_{BC}%
^{C^{\prime\prime}B^{\prime}}(\theta_{23})\,\Big|_{\theta_{12}=\frac{1}{2}%
i\pi}\label{Sb}\\%
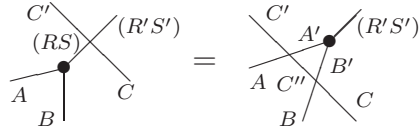
\begin{array}
[c]{c}%
\unitlength3.5mm\begin{picture}(6,5) \put(2,2){\line(1,1){2}} \put(4.5,1.5){\line(-1,1){3}} \put(2,2){\line(-4,-1){2}} \put(2,0){\line(0,1){2}} \put(.8,2.9){$_{(RS)}$} \put(4,3.5){$_{(R'S')}$} \put(0,.9){$_A$} \put(1,0){$_B$} \put(.6,4 ){$_{C'}$} \put(4,1){$_C$} \put(2,2){\makebox(0,0){$\bullet$}} \end{picture}
\end{array}
&  =~%
\begin{array}
[c]{c}%
\unitlength3.5mm\begin{picture}(6,5) \put(3,3){\line(1,1){1.2}} \put(4,0){\line(-1,1){4}} \put(0,2){\line(3,1){3}} \put(2,0){\line(1,3){1}} \put(4,3.5){$_{(R'S')}$} \put(1.8,3.2){$_{A'}$} \put(0,1.4){$_A$} \put(3,2){$_{B'}$} \put(1.1,0){$_B$} \put(.6,4 ){$_{C'}$} \put(1.,1.3 ){$_{C''}$} \put(4,.7){$_C$} \put(3,3){\makebox(0,0){$\bullet$}} \end{picture}
\end{array}
\nonumber
\end{align}
where $\theta_{(12)}=\frac{1}{2}\left(  \theta_{1}+\theta_{2}\right)  $ is the
bound state rapidity and $\theta_{12}/i=\pi/2$ the bound state fusion angle.
The bound state fusion intertwiner $\Gamma_{DE}^{(AB)}$ is defined by%
\begin{equation}
i\operatorname*{Res}_{\theta=\frac{1}{2}i\pi}S_{AB}^{B^{\prime}A^{\prime}%
}(\theta)=\sum_{R<S}\Gamma_{(RS)}^{B^{\prime}A^{\prime}}\Gamma_{AB}^{(RS)}=~~%
\begin{array}
[c]{c}%
\unitlength3.4mm\begin{picture}(3,6) \put(1,1){\oval(2,2)[t]} \put(1,5){\oval(2,2)[b]} \put(1,2){\makebox(0,0){$\bullet$}} \put(1,4){\makebox(0,0){$\bullet$}} \put(-.4,-.1){$A$} \put(1.5,-.1){$B$} \put(1.3,2.7){$(RS)$} \put(-.2,5.3){$B'$} \put(1.7,5.3){$A'$} \put(1,2){\line(0,1){2}} \end{picture}
\end{array}
~~. \label{inter}%
\end{equation}
With a convenient choice of an undetermined phase factor one has
\begin{equation}
\Gamma_{AB}^{(RS)}=i\Gamma(\tfrac{3}{4})\left(  2/\pi\right)  ^{1/4}\left(
\delta_{A}^{R}\delta_{B}^{S}-\delta_{A}^{S}\delta_{B}^{R}\right)  .
\label{inter1}%
\end{equation}
Applying formula (\ref{Sb}) twice we get the S-matrix for the scattering of
two bound states. For example we obtain%
\[
b^{SU(4)}b^{SU(4)}b^{SU(4)}b^{SU(4)}+b^{SU(4)}c^{SU(4)}b^{SU(4)}%
b^{SU(4)}-b^{SU(4)}c^{SU(4)}b^{SU(4)}c^{SU(4)}=-b^{O(6)}\left(  \theta\right)
\]
where the arguments on the left hand side are to be taken as $\theta+\frac
{1}{2}i\pi,\theta,\theta,\theta-\frac{1}{2}i\pi$. There are similar formulas
for the other amplitudes. The result is the S-matrix for the $O(6)$
Gross-Neveu model up to a sign\footnote{This is because the fundamental
Gross-Neveu particles are fermions (see \cite{KT1}).} (see \cite{KT1}).

We have the map $M_{(RS)}^{\alpha}$ from the anti-symmetric tensor $SU(4)$
bound states to the $O(6)$ vector states (in the complex basis) (see
\cite{KT1,BFK1,BFK8})%
\begin{equation}
\left.
\begin{array}
[c]{c}%
SU(4)~\text{bound states}\\

(RS)\in\{\left(  12\right)  ,\left(  13\right)  ,\left(  14\right)  ,\left(
23\right)  ,\left(  24\right)  ,\left(  34\right)  \}
\end{array}
\right\}  \leftrightarrow\left\{
\begin{array}
[c]{c}%
O(6)~\text{vector states}\\
\alpha\in\{1,2,3,\bar{3},\bar{2},\bar{1}\}
\end{array}
\right.  \,, \label{map}%
\end{equation}
which means that the no-zero matrix elements are%
\[
M_{(12)}^{1}=M_{(13)}^{2}=M_{(14)}^{3}=M_{(23)}^{\bar{3}}=M_{(24)}^{\bar{2}%
}=M_{(34)}^{\bar{1}}=1\,.
\]

\section{Form factors}

\label{s3}

The matrix element of a local operator $\mathcal{O}(x)$ for a state of $n$
particles of kind $\alpha_{i}$ with rapidities $\theta_{i}$
\begin{equation}
\langle\,0\,|\,\mathcal{O}(x)\,|\,\theta_{1},\dots,\theta_{n}\,\rangle
_{\underline{\alpha}}^{in}=e^{-ix(p_{1}+\cdots+p_{n})}F_{\underline{\alpha}%
}^{\mathcal{O}}(\underline{\theta})\,,~~\theta_{1}>\theta_{2}>\dots>\theta_{n}
\label{2.8}%
\end{equation}
defines the generalized form factor $F_{1\dots n}^{\mathcal{O}}({\underline
{\theta}})$, which is a co-vector valued function with components
$F_{\underline{\alpha}}^{\mathcal{O}}(\underline{\theta})$.

\paragraph{Form factor equations:}

The co-vector valued function $F_{1\dots n}^{\mathcal{O}}({\underline{\theta}%
})$ is meromorphic in all variables $\theta_{1},\dots,\theta_{n}$ and
satisfies the following relations \cite{KW,Sm}:

\begin{itemize}
\item[(i)] The Watson's equations describe the symmetry property under the
permutation of both, the variables $\theta_{i},\theta_{j}$ and the spaces
$i,j=i+1$ at the same time
\begin{equation}
F_{\dots ij\dots}^{\mathcal{O}}(\dots,\theta_{i},\theta_{j},\dots)=F_{\dots
ji\dots}^{\mathcal{O}}(\dots,\theta_{j},\theta_{i},\dots)\,S_{ij}(\theta_{ij})
\label{1.10}%
\end{equation}
for all possible arrangements of the $\theta$'s.

\item[(ii)] The crossing relation implies a periodicity property under the
cyclic permutation of the rapidity variables and spaces
\begin{multline}
^{~\text{out,}\bar{1}}\langle\,p_{1}\,|\,\mathcal{O}(0)\,|\,p_{2},\dots
,p_{n}\,\rangle_{2\dots n}^{\text{in,conn.}}\\
=F_{1\ldots n}^{\mathcal{O}}(\theta_{1}+i\pi,\theta_{2},\dots,\theta
_{n})\sigma_{1}^{\mathcal{O}}\mathbf{C}^{\bar{1}1}=F_{2\ldots n1}%
^{\mathcal{O}}(\theta_{2},\dots,\theta_{n},\theta_{1}-i\pi)\mathbf{C}%
^{1\bar{1}} \label{1.12}%
\end{multline}
where $\sigma_{\alpha}^{\mathcal{O}}$ takes into account the statistics of the
particle $\alpha$ with respect to $\mathcal{O}$. The charge conjugation matrix
$\mathbf{C}^{\bar{1}1}$ will be discussed below.

\item[(iii)] There are poles determined by one-particle states in each
sub-channel given by a subset of particles of the state in (\ref{2.8}). In
particular the function $F_{\underline{\alpha}}^{\mathcal{O}}({\underline
{\theta}})$ has a pole at $\theta_{12}=i\pi$ such that
\begin{equation}
\operatorname*{Res}_{\theta_{12}=i\pi}F_{1\dots n}^{\mathcal{O}}(\theta
_{1},\dots,\theta_{n})=2i\,\mathbf{C}_{12}\,F_{3\dots n}^{\mathcal{O}}%
(\theta_{3},\dots,\theta_{n})\left(  \mathbf{1}-\sigma_{2}^{\mathcal{O}}%
S_{2n}\dots S_{23}\right)  \,. \label{1.14}%
\end{equation}

\item[(iv)] If there are also bound states in the model the function
$F_{\underline{\alpha}}^{\mathcal{O}}({\underline{\theta}})$ has additional
poles. If for instance the particles 1 and 2 form a bound state (12), there is
a pole at $\theta_{12}=i\eta$ such that
\begin{equation}
\operatorname*{Res}_{\theta_{12}=i\eta}F_{12\dots n}^{\mathcal{O}}(\theta
_{1},\theta_{2},\dots,\theta_{n})\,=F_{(12)\dots n}^{\mathcal{O}}%
(\theta_{(12)},\dots,\theta_{n})\,\sqrt{2}\Gamma_{12}^{(12)} \label{1.16}%
\end{equation}
where the bound state intertwiner $\Gamma_{12}^{(12)}$ is here given by
(\ref{inter1}) and the values of $\theta_{1},\,\theta_{2},\,\theta_{(12)}$ are
given in general in \cite{K1,KT1,BK}.

\item[(v)] Naturally, since we are dealing with relativistic quantum field
theories we finally have
\begin{equation}
F_{1\dots n}^{\mathcal{O}}(\theta_{1}+\mu,\dots,\theta_{n}+\mu)=e^{s\mu
}\,F_{1\dots n}^{\mathcal{O}}(\theta_{1},\dots,\theta_{n}) \label{1.18}%
\end{equation}
if the local operator transforms under Lorentz transformations as
$\mathcal{O}\rightarrow e^{s\mu}\mathcal{O}$ where $s$ is the
\textquotedblleft spin\textquotedblright\ of $\mathcal{O}$.
\end{itemize}

\noindent For the $SU(4)$ S-matrix (\ref{S4}) the bound state pole is at
$\theta=\frac{1}{2}i\pi$, i.e. $\eta=\frac{1}{2}$.

\paragraph{The general form factor formula:}

We write the general form factor $F_{1\dots n}^{\mathcal{O}}(\underline
{\theta})$ for $n$ fundamental particles following \cite{KW} as
\begin{equation}
F_{\underline{\alpha}}^{\mathcal{O}}(\underline{\theta})=K_{\underline{\alpha
}}^{\mathcal{O}}(\underline{\theta})\prod_{1\leq i<j\leq n}F(\theta_{ij})
\label{FK}%
\end{equation}
where $F(\theta)$ is the minimal form factor function (see below). The
K-function $K_{1\dots n}^{\mathcal{O}}(\underline{\theta})$ is given in terms
of a nested `off-shell' Bethe ansatz (see e.g. \cite{BFK7,BFK1})
\begin{equation}
\fbox{$\rule[-0.2in]{0in}{0.5in}\displaystyle~K_{\underline{\alpha}%
}^{\mathcal{O}}(\underline{\theta})=\int_{\mathcal{C}_{\underline{\theta}%
}^{(1)}}dz_{1}\cdots\int_{\mathcal{C}_{\underline{\theta}}^{(m)}}%
dz_{m}\,h(\underline{\theta},\underline{z})p^{\mathcal{O}}(\underline{\theta
},\underline{z})\,\tilde{\Psi}_{\underline{\alpha}}(\underline{\theta
},\underline{z})$~} \label{K}%
\end{equation}
written as a multiple contour integral. The scalar function $h(\underline
{\theta},\underline{z})$ depends only on the S-matrix and not on the specific
operator $\mathcal{O}(x)$%
\begin{align}
h(\underline{\theta},\underline{z})  &  =\prod_{i=1}^{n}\prod_{j=1}^{m}%
\tilde{\phi}(\theta_{i}-z_{j})\prod_{1\leq i<j\leq m}\tau(z_{i}-z_{j}%
)\label{h}\\
\tau(z)  &  =\frac{1}{\tilde{\phi}(-z)\tilde{\phi}(z)}\,. \label{tau}%
\end{align}
The dependence on the specific operator $\mathcal{O}(x)$ is encoded in the
scalar p-function $p^{\mathcal{O}}(\underline{\theta},\underline{z})$ which is
in general a simple function of $e^{\theta_{i}}$ and $e^{z_{j}}$. The function
$\tilde{\phi}(\theta)$ and the integration contours $\mathcal{C}%
_{\underline{\theta}}$ depend on the model and are given below.

\subsection{$SU(4)$ form factors}

\paragraph{Minimal form factor:}

Let $S\left(  \theta\right)  $ be an S-matrix eigenvalue. The solution of
Watson's and the crossing equations (i) and (ii) for two particles%
\begin{equation}%
\begin{array}
[c]{c}%
F\left(  \theta\right)  =S\left(  \theta\right)  F\left(  -\theta\right) \\
F\left(  i\pi+\theta\right)  =F\left(  i\pi-\theta\right)
\end{array}
\label{wat}%
\end{equation}
with no poles in the physical strip $0\leq\operatorname{Im}\theta\leq\pi$ and
at most a simple zero at $\theta=0$ is the minimal form factor \cite{KW}. For
the construction of the `off-shell' Bethe ansatz the minimal form factor for
the highest weight eigenvalue of the $SU(4)$ S-matrix $a^{SU(4)}(\theta)$ of
(\ref{a4}) is essential. The unique solution (up to a constant factor) is%
\begin{equation}
F^{SU(4)}\left(  \theta\right)  =\frac{G\left(  \frac{1}{2}\frac{\theta}{i\pi
}\right)  G\left(  1-\frac{1}{2}\frac{\theta}{i\pi}\right)  }{G\left(
\frac{3}{4}+\frac{1}{2}\frac{\theta}{i\pi}\right)  G\left(  \frac{7}{4}%
-\frac{1}{2}\frac{\theta}{i\pi}\right)  } \label{F4}%
\end{equation}
where $G\left(  z\right)  $ is Barnes G-function, which satisfies (see e.g.
\cite{Wo})%
\[
G\left(  1+z\right)  =\Gamma\left(  z\right)  G\left(  z\right)  .
\]

\paragraph{The $n$-particle form factor for $SU(4)$}

is given by (\ref{K}) and the function $\tilde{\phi}(\theta)$ in (\ref{h}) and
(\ref{tau}) is (see \cite{BFK1})%
\begin{equation}
\tilde{\phi}(\theta)=\Gamma\left(  \tfrac{3}{4}+\tfrac{1}{2\pi i}%
\theta\right)  \Gamma\left(  -\tfrac{1}{2\pi i}\theta\right)  . \label{phi4}%
\end{equation}
The integration contour in (\ref{K}) for $SU(4)$ is depicted in Fig.
\ref{f2}\begin{figure}[th]%
\[
\unitlength3.5mm\begin{picture}(25,13)
\thicklines
\put(1,0){
\put(0,0){$\bullet~\theta_n- 2i\pi$}
\put(.19,4.5){\circle{.3}~$\theta_n-\frac12 i\pi$}
\put(0,6){$\bullet~~\theta_n$}
\put(.3,6.3){\oval(1,1)}
\put(.68,6.78){\vector(1,0){0}}
\put(0,10.5){$\bullet~\theta_n+\frac32 i\pi$}
}
\put(8,6){\dots}
\put(12,0){
\put(0,0){$\bullet~\theta_2- 2i\pi$}
\put(.19,4.5){\circle{.3}~$\theta_n-\frac12 i\pi$}
\put(0,6){$\bullet~~\theta_2$}
\put(.3,6.3){\oval(1,1)}
\put(.68,6.78){\vector(1,0){0}}
\put(0,10.5){$\bullet~\theta_2+\frac32 i\pi$}
}
\put(20,1){
\put(0,0){$\bullet~\theta_1- 2i\pi$}
\put(.19,4.5){\circle{.3}~$\theta_1-\frac12 i\pi$}
\put(0,6){$\bullet~~\theta_1$}
\put(.3,6.3){\oval(1,1)}
\put(.68,6.78){\vector(1,0){0}}
\put(0,10.5){$\bullet~\theta_1+\frac32 i\pi$}
}
\put(9,2){\vector(1,0){0}}
\put(0,2.5){\oval(34,1)[br]}
\put(25,2.5){\oval(16,1)[tl]}
\end{picture}
\]
\caption{The integration contour for $SU(4)$ }%
\label{f2}%
\end{figure}

\paragraph{Nesting:}

The Bethe state in (\ref{K}) for $SU(4)$ is written as%
\begin{equation}
\tilde{\Psi}_{\underline{A}}(\underline{\theta},\underline{z})=K_{\underline
{B}}^{(1)}(\underline{z})\,\tilde{\Phi}_{\underline{A}}^{\underline{B}%
}(\underline{\theta},\underline{z}) \label{1.37}%
\end{equation}
where $\underline{A}=\left(  A_{1},\dots,A_{n}\right)  $ with $1\leq A_{i}%
\leq4$ and summation over all $\underline{B}=(B_{1},\dots,B_{m})$ with $2\leq
B_{i}\leq4$ is assumed. The basic Bethe ansatz co-vectors (in the algebraic
formulation) $\tilde{\Phi}_{1\dots n}^{\underline{B}}\in\left(  V^{1\dots
n}\right)  ^{\dag}$ are defined as \cite{BFK1}%
\begin{equation}%
\begin{array}
[c]{rcl}%
\tilde{\Phi}_{1\dots n}^{\underline{B}}(\underline{\theta},\underline{z}) &
= & \Omega_{1\dots n}\tilde{C}_{1\dots n}^{B_{m}}(\underline{\theta}%
,z_{m})\cdots\tilde{C}_{1\dots n}^{B_{1}}(\underline{\theta},z_{1})\\
\tilde{\Phi}_{\underline{A}}^{\underline{B}}(\underline{\theta},\underline
{z}) & = &
\begin{array}
[c]{c}%
\unitlength4mm\begin{picture}(9,7) \put(9,5){\oval(14,2)[lb]} \put(9,5){\oval(18,6)[lb]} \put(4,1){\line(0,1){4}} \put(8,1){\line(0,1){4}} \put(-.2,5.4){$B_1$} \put(1.8,5.4){$B_m$} \put(3.8,.1){$A_1$} \put(7.8,.1){$A_n$} \put(3.8,5.4){$1$} \put(7.8,5.4){$1$} \put(9.2,1.8){$1$} \put(9.2,3.8){$1$} \put(3,2.5){$\theta_1$} \put(6.8,2.5){$\theta_{n}$} \put(.8,2.5){$z_1$} \put(1.7,3.6){$z_m$} \put(5.4,4.5){$\dots$} \put(8.5,2.6){$\vdots$} \end{picture}
\end{array}
~~~,~~~%
\begin{array}
[c]{l}%
2\leq B_{i}\leq4\\
1\leq A_{i}\leq4~.
\end{array}
\end{array}
\label{1.38}%
\end{equation}
The nested Bethe ansatz is obtained by writing for $K_{\underline{B}}%
^{(1)}(\underline{z})$ of (\ref{1.37}) an ansatz as (\ref{K}) and so on: for
$K_{\underline{B}}^{(1)}(\underline{z}^{(1)})$ we have an $SU(3)$ and for
$K_{\underline{C}}^{(2)}(\underline{z}^{(2)})$ an $SU(2)$ Bethe ansatz, which
is well known. The number $m=n_{1}$ in (\ref{1.38}) is the number of
\textquotedblleft weight flip\textquotedblright operators. These numbers for
the various levels of the nested Bethe ansatz satisfy \cite{BFK1}%
\begin{equation}
(n-n_{1},n_{1}-n_{2},n_{2}-n_{3},n_{3})=w^{\emph{O}}+L(1,1,1,1) \label{w}%
\end{equation}
where $w^{\emph{O}}$ is the weight vector of the operator $\mathcal{O}$ and
$L=0,1,2,\dots$.

\subsection{$O(6)$ form factors}

\subparagraph{Minimal form factors:}

The solutions of Watson's and the crossing equations (i) and (ii) for two
particles (\ref{wat}) with no poles in the physical strip $0\leq
\operatorname{Im}\theta\leq\pi$ and at most a simple zero at $\theta=0$ are
the minimal form factors \cite{KW}%
\begin{equation}
\left(  F_{0}^{O(6)},F_{+}^{O(6)},F_{-}^{O(6)}\right)  ^{\min}=\left(
\frac{2\tanh\frac{1}{2}\left(  i\pi-\theta\right)  }{i\pi-\theta},\frac
{\Gamma\left(  \frac{5}{4}-\frac{1}{2\pi i}\theta\right)  \Gamma\left(
\frac{1}{4}+\frac{1}{2\pi i}\theta\right)  }{\Gamma^{2}\left(  \frac{3}%
{4}\right)  \cosh\frac{1}{2}\left(  i\pi-\theta\right)  },1\right)
F_{-}^{O(6)\min}\,. \label{Fmin6}%
\end{equation}
They belong to the S-matrix eigenvalues $S_{0}^{O(6)}$ and $S_{\pm}^{O(6)}$ of
(\ref{SO6}). The full 2-particle form factors are%
\begin{equation}
F_{+,-,0}^{O(6)}\left(  \theta\right)  =\frac{1}{\sinh\frac{1}{2}(\theta
-\frac{1}{2}i\pi)\sinh\frac{1}{2}(\theta+\frac{1}{2}i\pi)}F_{+,-,0}^{O(6)\min
}\left(  \theta\right)  . \label{Fpm06}%
\end{equation}
They are non-minimal solutions of (\ref{wat}) having a pole at $\theta
=\frac{1}{2}i\pi$ (see (5.10) and (2.16) of \cite{KW}). For the construction
of the `off-shell' Bethe ansatz the minimal solution of the form factor
equation (\ref{wat}) for the highest weight eigenvalue of the $O(N)$
S-matrix\footnote{The minus sign in (\ref{Fmin0}) is due to fermionic
statistics of the fundamental particles (see also eq. 4.12 of \cite{BFKZ}).}
\begin{equation}
F^{O(6)}\left(  \theta\right)  =-a^{O(6)}(\theta)F^{O(6)}\left(
-\theta\right)  \label{Fmin0}%
\end{equation}
is essential. The unique solution (up to a constant factor) is%
\begin{align}
F^{O(6)}\left(  \theta\right)   &  =c\cosh\tfrac{1}{2}\left(  i\pi
-\theta\right)  \,F_{+}^{O(6)\min}\left(  \theta\right) \label{F6}\\
&  =\frac{G\left(  \frac{1}{2}\frac{\theta}{i\pi}\right)  G\left(  1-\frac
{1}{2}\frac{\theta}{i\pi}\right)  }{G\left(  \frac{1}{2}+\frac{1}{2}%
\frac{\theta}{i\pi}\right)  G\left(  \frac{3}{2}-\frac{1}{2}\frac{\theta}%
{i\pi}\right)  }\frac{G\left(  \frac{1}{4}+\frac{1}{2}\frac{\theta}{i\pi
}\right)  G\left(  \frac{5}{4}-\frac{1}{2}\frac{\theta}{i\pi}\right)
}{G\left(  \frac{3}{4}+\frac{1}{2}\frac{\theta}{i\pi}\right)  G\left(
\frac{7}{4}-\frac{1}{2}\frac{\theta}{i\pi}\right)  }.\nonumber
\end{align}
The function $\,\tilde{\phi}(\theta)$ in (\ref{h}) is the same as (\ref{phi4})
for $SU(4)$ and the integration contours in (\ref{K}) can be found in
\cite{BFK8}.

\subsection{Bound state form factors}

The statistics factor of two fundamental particles in the chiral $SU(N)$
Gross-Neveu model \cite{KuS,KKS,BFK1} is $\sigma=\exp\left(  2\pi is\right)
$, where $s=\frac{1}{2}\left(  1-\frac{1}{N}\right)  $ is the spin. For
$SU(4)$ this means that the spin is $s=\frac{3}{8},$ and the statistics factor
is$~\sigma=\exp\left(  \frac{3}{4}\pi i\right)  $. In particular, the bound
states of two fundamental $SU(4)$ particles are fermions because $\sigma
^{4}=-1$.

An $n^{\prime}=n/2$-particle form factor for $O(6)$ is calculated from an
$n$-particle one of $SU(4)$ using the bound state formula (iv) of (\ref{1.16})%
\begin{equation}
F_{\underline{\alpha}}^{O(6)}(\underline{\omega})\Gamma_{\underline{A}%
}^{\underline{\alpha}}=2^{-n/4}\operatorname*{Res}_{\theta_{12}=\frac{1}%
{2}i\pi}\ldots\operatorname*{Res}_{\theta_{n-1n}=\frac{1}{2}i\pi}%
F_{\underline{A}}^{SU(4)}(\underline{\theta}) \label{Fn}%
\end{equation}
where $\Gamma_{\underline{A}}^{\underline{\alpha}}=\Gamma_{A_{1}A_{2}}%
^{\alpha_{1}}\ldots\Gamma_{A_{n-1}A_{n}}^{\alpha_{n^{\prime}}}$ is the total
intertwiner and $\omega_{i}=\frac{1}{2}\left(  \theta_{2i-1}+\theta
_{2i}\right)  $ are the bound state rapidities.

\begin{lemma}
\label{l1}The bound state form factors defined by (\ref{Fn}) satisfy the form
factor equations (i) - (v) of (\ref{1.10}) - (\ref{1.16}). The K-functions
defined by (\ref{FK}) and (\ref{Fn}) satisfy, in particular

\begin{enumerate}
\item
\begin{equation}
K_{\underline{\alpha}}^{O(6)}(\underline{\omega})\Gamma_{\underline{A}%
}^{\underline{\alpha}}=2^{-n/4}\prod\limits_{1\leq i<j\leq n^{\prime}}\frac
{1}{\tilde{\phi}(-\omega_{ij})\tilde{\phi}(-\omega_{ij}+\frac{1}{2}i\pi
)}\operatorname*{Res}_{\theta_{12}=\frac{1}{2}i\pi}\ldots\operatorname*{Res}%
_{\theta_{n-1n}=\frac{1}{2}i\pi}K_{\underline{A}}^{SU(4)}(\underline{\theta})
\label{Kn}%
\end{equation}

\item the form factor equation (iii) in the form (see \cite{BFK8})%
\begin{equation}
\operatorname*{Res}_{\omega_{12}=i\pi}K_{1\dots n^{\prime}}^{O(6)}%
(\underline{\omega})=\frac{2i}{F^{O(6)}(i\pi)}\mathbf{C}_{12}\prod
_{i=3}^{n^{\prime}}\tilde{\phi}(\omega_{i1}+\tfrac{1}{2}i\pi)\tilde{\phi
}(\omega_{i2})K_{3\dots n^{\prime}}^{O(6)}(\underline{\check{\omega}})\left(
\mathbf{1}-S_{2n^{\prime}}\dots S_{23}\right)  \label{iii6}%
\end{equation}
with $\underline{\check{\omega}}=\omega_{3},\ldots\omega_{n^{\prime}}$.
\end{enumerate}
\end{lemma}

\begin{proof}
In Appendix E of \cite{BK} was proved that in general bound state form factors
satisfy the form factor equations. We use the variables $u,o$ with
$\theta=\tfrac{1}{2}i\pi u,~\omega=\tfrac{1}{2}i\pi o$.

\begin{enumerate}
\item Equation (\ref{Fn}) implies for the K-functions (\ref{Kn}) because from
(\ref{F4}) and (\ref{F6}) we derive%
\begin{equation}
\frac{F^{SU(4)}\left(  \omega+\frac{1}{2}i\pi\right)  \left(  F^{SU(4)}\left(
\omega\right)  \right)  ^{2}F^{SU(4)}\left(  \omega-\frac{1}{2}i\pi\right)
}{F^{O(6)}\left(  \omega\right)  }=\frac{1}{\tilde{\phi}(-\omega)\tilde{\phi
}(-\omega+\frac{1}{2}i\pi)}. \label{F46}%
\end{equation}

\item This follows from the general proof of (iii) in \cite{BK} and
(\ref{Kn}). One can also prove it directly from (iii) for $F_{\underline{A}%
}^{SU(4)}(\underline{\theta})$, eq. (\ref{Kn}) and (up to a const.)%
\[
\frac{\prod\limits_{3\leq i<j\leq n^{\prime}}\tilde{\phi}(-\omega_{ij}%
)\tilde{\phi}(-\omega_{ij}+\frac{1}{2}i\pi)}{\prod\limits_{1\leq i<j\leq
n^{\prime}}\tilde{\phi}(-\omega_{ij})\tilde{\phi}(-\omega_{ij}+\frac{1}{2}%
i\pi)}\prod_{i=5}^{n}\prod_{j=2}^{4}\tilde{\phi}(\theta_{ij})=\prod
_{i=3}^{n^{\prime}}\tilde{\phi}(\omega_{i1}+\tfrac{1}{2}i\pi)\tilde{\phi
}(\omega_{i2})
\]
for $\theta_{2i-1}=\omega_{i}+\frac{1}{4}i\pi,~\theta_{2i}=\omega_{i}-\frac
{1}{4}i\pi$ and $\omega_{12}=i\pi$.
\end{enumerate}
\end{proof}

This lemma implies the following

\begin{corollary}
In \cite{KT1} we demonstrated that the isomorphism $O(6)\simeq SU(4)$, leads
to an equivalence between the $O(6)$ Gross-Neveu model and the $SU(4)$ chiral
Gross-Neveu model for the S-matrices. The results of this section show, that
this is also true for the form factors.
\end{corollary}

\section{Examples of operators}

\label{s4}

We use the results of \cite{BFK1} and \cite{BFK8}.

\subsection{The current $J^{\mu}(x)$}

\paragraph{The $SU(4)$ form factor:}

The $SU(4)$ Noether current $J_{A\bar{B}}^{\mu}(x)$ transforms as the adjoint
representation with weight vector $w^{J}=(2,1,1,0)$. Because the Bethe ansatz
yields highest weight states we consider the highest weight component%
\[
J_{1\bar{4}}^{\mu}(x)=\epsilon^{\mu\nu}\partial_{\nu}J(x)\,,
\]
where the anti-particle $\bar{4}$ is defined by (\ref{C4}) and $J(x)$ is the
pseudo potential with the p-function in (\ref{K}) (see subsection 4.3 of
\cite{BFK1})
\begin{equation}
p^{J}(\underline{\theta},\underline{\underline{z}})=e^{\frac{1}{2}\left(
\sum\theta_{i}-\sum z_{i}^{(1)}-\sum z_{i}^{(3)}\right)  }/\sum e^{\theta_{i}%
}.\label{pJ}%
\end{equation}
The $n$-particle current form factor for $SU(4)$ is given by (\ref{FK}) and
the nested `off-shell' Bethe ansatz (\ref{K}) with the p-function (\ref{pJ}).
The numbers of \textquotedblleft weight flip\textquotedblright\ operators in
the various levels of the nested Bethe ansatz are given by (\ref{w}) as
$n=4+4L,\,n_{1}=2+3L,\,n_{2}=1+2L,\,n_{3}=L$. In particular we consider $L=0$,
i.e. $n=4,\,n_{1}=2,\,n_{2}=1$ and $n_{3}=0$. The Bethe state in (\ref{K}) is
then%
\begin{align}
\Psi_{\underline{A}}(\underline{\theta},\underline{z}) &  =K_{\underline{B}%
}^{(1)}(\underline{z})\Phi_{\underline{A}}^{\underline{B}}(\underline{\theta
},\underline{z})\nonumber\\
K_{\underline{B}}^{(1)}(\underline{z}) &  =\int dy\prod_{j=1}^{2}\tilde{\phi
}(z_{j}-y)\Psi_{\underline{B}}^{(1)}(\underline{z},y)\label{K1}\\
\Psi_{\underline{B}}^{(1)}(\underline{z},y) &  =\delta_{B_{1}}^{2}%
\delta_{B_{2}}^{3}\tilde{b}(z_{1}-y)\tilde{c}(z_{2}-y)+\delta_{B_{1}}%
^{3}\delta_{B_{2}}^{2}\tilde{c}(z_{1}-y)\nonumber
\end{align}
(see also Fig. \ref{fj}). Below we use this formula to calculate the bound
state form factor.

\paragraph{The $O(6)$ form factor:}

The $O(6)$ Noether current transforms as an antisymmetric $O(N)$ tensor with
weights $w^{J}=(1,1,0)$. The bound state formula (\ref{Fn}) applied to
(\ref{FK},\ref{K}) with the p-function (\ref{pJ}) yields the $O(6)$ current
form factor for $n/2$ particles. In particular we consider the case $n=4$.

\begin{proposition}
The bound state formula (\ref{Fn}) for $n=4$ and (\ref{FK}, \ref{K}) with the
p-function (\ref{pJ}) yield the two particle $O(6)$ form factor of the
pseudo-potential $J^{\alpha\beta}(x)$ and the current $J_{\mu}^{\alpha\beta
}(x)=\epsilon_{\mu\nu}\partial^{\nu}J^{\alpha\beta}(x)$%
\begin{align}
F_{\alpha_{1}\alpha_{2}}^{O(6),J^{\alpha\beta}}\left(  \theta_{1},\theta
_{2}\right)   &  =im\left(  \delta_{\alpha_{1}}^{\alpha}\delta_{\alpha_{2}%
}^{\beta}-\delta_{\alpha_{1}}^{\beta}\delta_{\alpha_{2}}^{\alpha}\right)
\frac{1}{\cosh\frac{1}{2}\theta_{12}}F_{-}^{O(6)}\left(  \theta\right)
\label{FJ}\\
F_{\alpha_{1}\alpha_{2}}^{O(6),J_{\mu}^{\alpha\beta}}\left(  \theta_{1}%
,\theta_{2}\right)   &  =i\left(  \delta_{\alpha_{1}}^{\alpha}\delta
_{\alpha_{2}}^{\beta}-\delta_{\alpha_{1}}^{\beta}\delta_{\alpha_{2}}^{\alpha
}\right)  \bar{v}(\theta_{1})\gamma_{\mu}u(\theta_{2})F_{-}^{O(6)}\left(
\theta\right)  \label{FJmu}%
\end{align}
which agrees with the results of \cite{BFK8}.
\end{proposition}

\begin{proof}
We have $n=4,\,n_{1}=2,\,n_{2}=1$ and $n_{3}=0$. For convenience we use here
the variables $u,v,w$ with $\theta=i\pi\frac{1}{2}u,~z=i\pi\frac{1}%
{2}v,~y=i\pi\frac{1}{2}w$ and calculate (always up to constants)%
\begin{align}
\operatorname*{Res}_{u_{12}=1}\operatorname*{Res}_{u_{34}=1}K_{\underline{A}%
}^{SU(4),J}(\underline{u}) &  =\operatorname*{Res}_{u_{12}=1}%
\operatorname*{Res}_{u_{34}=1}\int_{\mathcal{C}_{\underline{u}}}d\underline
{v}h\left(  \underline{u},\underline{v}\right)  p^{J}(\underline{u}%
,\underline{v})\Psi_{\underline{A}}(\underline{u},\underline{v})\label{Kj}\\
&  =\operatorname*{Res}_{u_{12}=1}\operatorname*{Res}_{u_{34}=1}%
\operatorname*{Res}_{v_{1}=u_{2}}\operatorname*{Res}_{v_{2}=u_{4}}h\left(
\underline{u},\underline{v}\right)  p^{J}(\underline{u},\underline{v}%
)\Psi_{\underline{A}}(\underline{u},\underline{v})\nonumber
\end{align}
because the residues are obtained by pinchings at $v_{1}=u_{2},v_{2}=u_{4}$
which imply that the S-matrices $S(u_{2}-v_{1})$ and $S(u_{4}-v_{2})$ are
replaced by the permutation operator (see Fig. \ref{fj}). 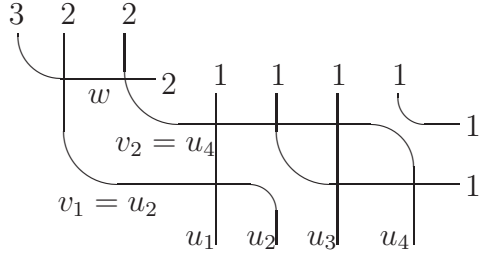
\begin{figure}[th]%
\[
\unitlength4mm\begin{picture}(15,8.5)(13.5,2)
\put(20,2){\line(0,1){5}}
\put(21,2){\oval(2,4)[rt]}
\put(25,2){\oval(3,8)[rt]}
\put(25,9){\oval(16,6)[lb]}
\put(21,9){\oval(12,10)[lb]}
\put(18,9){\oval(9,3)[lb]}
\put(24,2){\line(0,1){5}}
\put(28,7){\oval(12,6)[lb]}
\put(28,7){\oval(4,2)[lb]}
\put(19,2){$u_1$}
\put(21,2){$u_2$}
\put(23,2){$u_3$}
\put(25.4,2){$u_4$}
\put(14.8,3.1){$v_1=u_2$}
\put(16.7,5.1){$v_2=u_4$}
\put(15.8,6.7){$w$}
\put(14.9,9.3){2}
\put(16.9,9.3){2}
\put(18.2,7.){2}
\put(13.2,9.3){3}
\put(28.2,3.6){1}
\put(28.2,5.6){1}
\put(19.9,7.3){1}
\put(21.8,7.3){1}
\put(23.8,7.3){1}
\put(25.8,7.3){1}
\end{picture}~
\]
\caption{The Bethe state $\Psi_{\underline{A}}(\underline{u},\underline{v})$
in (\ref{Kj}) for $v_{1}=u_{2},v_{2}=u_{4}$.}%
\label{fj}%
\end{figure}Using Yang-Baxter relations and the formula for the fusion
intertwiner (\ref{inter}) we obtain%
\[
\operatorname*{Res}_{u_{12}=1}\operatorname*{Res}_{u_{34}=1}K_{\underline{A}%
}^{SU(4),J}(\underline{u})=\tilde{\phi}(u_{14})\tilde{\phi}(u_{32}%
)K_{\underline{B}}^{(1)}(u_{24})\tilde{b}(u_{14})p^{J}\left(  \underline
{u},u_{2},u_{4}\right)  \left(  \Gamma_{\alpha}^{B_{1}1}\Gamma_{A_{1}A_{2}%
}^{\alpha}\Gamma_{\beta}^{B_{2}1}\Gamma_{A_{3}A_{4}}^{\beta}\right)  .
\]
With (\ref{K1}) we have (again up to constants)%
\begin{align*}
K_{32}^{(1)}(\underline{v}) &  =\int_{\mathcal{C}_{\underline{v}}}%
dw\Gamma(-\tfrac{1}{4}(v_{1}-w))\Gamma(-\tfrac{1}{4}+\tfrac{1}{4}%
(v_{1}-w))\Gamma(-\tfrac{1}{4}(v_{2}-w))\Gamma(\tfrac{3}{4}+\tfrac{1}{4}%
(v_{2}-w))\\
&  =\Gamma\left(  \tfrac{3}{4}-\tfrac{1}{4}v_{12}\right)  \Gamma\left(
-\tfrac{1}{4}+\tfrac{1}{4}v_{12}\right)
\end{align*}
where $\tilde{\phi}(v_{1}-w)\tilde{c}(v_{1}-w)\propto\Gamma(-\tfrac{1}%
{4}(v_{1}-w))\Gamma(-\tfrac{1}{4}+\tfrac{1}{4}(v_{1}-w))$ and the Gauss
formula%
\begin{equation}
_{2}F_{1}(a,b;c;1)=\sum_{n=0}^{\infty}\frac{1}{n!}\frac{\Gamma\left(
a+n\right)  }{\Gamma\left(  a\right)  }\frac{\Gamma\left(  b+n\right)
}{\Gamma\left(  b\right)  }\frac{\Gamma\left(  c\right)  }{\Gamma\left(
c+n\right)  }=\frac{\Gamma\left(  c\right)  \Gamma\left(  c-a-b\right)
}{\Gamma\left(  c-a\right)  \Gamma\left(  c-b\right)  }\label{Gauss}%
\end{equation}
have been used. Similarly, we calculate $K_{23}^{(1)}(\underline{v})$ and get
$K_{23}^{(1)}(\underline{v})=-K_{32}^{(1)}(\underline{v})$. Finally using
(\ref{pJ})
\begin{align*}
&  \operatorname*{Res}_{u_{12}=1}\operatorname*{Res}_{u_{34}=1}K_{\underline
{A}}^{SU(4),J}(\underline{u})=K^{J}(o)\left(  \Gamma_{A_{1}A_{2}}^{(12)}%
\Gamma_{A_{3}A_{4}}^{(13)}-\Gamma_{A_{1}A_{2}}^{(13)}\Gamma_{A_{3}A_{4}%
}^{(12)}\right)  \\
K^{J}(o) &  =\tilde{\phi}(u_{14})\tilde{\phi}(u_{32})\Gamma\left(  \tfrac
{3}{4}-\tfrac{1}{4}u_{24}\right)  \Gamma\left(  -\tfrac{1}{4}+\tfrac{1}%
{4}u_{24}\right)  \tilde{b}(u_{14})p^{J}\left(  \underline{u},u_{2}%
,u_{4}\right)  \\
&  =\frac{1}{\sin\frac{1}{2}\pi o}\left(  \Gamma\left(  \tfrac{3}{4}-\tfrac
{1}{4}o\right)  \Gamma\left(  -\tfrac{1}{4}+\tfrac{1}{4}o\right)  \right)
^{2}%
\end{align*}
with $o=u_{(12)(34)}=u_{13}=u_{24}=u_{14}-\tfrac{1}{2}=u_{23}+\tfrac{1}{2}$.
The result (\ref{FJ}) follows then from (\ref{F46}), (\ref{Fmin6}) and
(\ref{Fpm06}).
\end{proof}

\subsection{The iso-scalar operator $\mathcal{O}$}

\label{sscalar}

The $SU(4)~n$-particle form factor for the iso-scalar operator $\mathcal{O}%
(x)$ with weights $w^{\mathcal{O}}=\left(  0,0,0,0\right)  $ is given by
(\ref{FK}) and the nested `off-shell' Bethe ansatz (\ref{K}). The numbers of
\textquotedblleft weight flip\textquotedblright\ operators in the various
levels of the nested Bethe ansatz are given by (\ref{w}) as $n=4+4L,\,n_{1}%
=3+3L,\,n_{2}=2+2L,\,n_{3}=1+L$. We propose for the iso-scalar operator
$\mathcal{O}(x)$ the p-function
\begin{equation}
p^{\mathcal{O}}\left(  \underline{\theta},\underline{\underline{z}}\right)
=e^{\frac{1}{2}\sum\theta_{i}-\sum z_{i}^{(1)}+\sum z_{i}^{(3)}}-1. \label{ps}%
\end{equation}
With this p-function in (\ref{K}) the form factor equations (i) - (v) of
(\ref{1.10}) - (\ref{1.16}) hold with statistics factor $\sigma^{\mathcal{O}%
}=-1$ and spin $s^{\mathcal{O}}=0$. The bound state formula (\ref{Fn}) applied
to (\ref{FK}) and (\ref{K}) with the p-function (\ref{ps}) yields the $O(6)$
form factor of the operator $\bar{\psi}\psi(x)$ for $n/2$ particles. In
particular we consider the case $L=0$, i.e $n=4,\,n_{1}=3,\,n_{2}=2$ and
$n_{3}=1$.

\begin{proposition}
The bound state formula (\ref{Fn}) applied to (\ref{FK}) and (\ref{K}) with
the p-function (\ref{ps}) yields the two particle $O(6)$ form factors of
$\bar{\psi}\psi$%
\begin{equation}
F_{\alpha_{1}\alpha_{2}}^{\bar{\psi}\psi}(\underline{\theta})=\langle
\,0\,|\,\bar{\psi}\psi(0)\,|\,p_{1},p_{2}\,\rangle_{\alpha_{1}\alpha_{2}}%
^{in}=\mathbf{C}_{\alpha_{1}\alpha_{2}}\,\bar{v}(\theta_{1})u(\theta
_{2})\,F_{0}(\theta_{12}) \label{EMN}%
\end{equation}
which means for the energy momentum operator $T^{\mu\nu}$%
\[
F_{\alpha_{1}\alpha_{2}}^{T^{\mu\nu}}(\underline{\theta})=\langle
\,0\,|\,T^{\mu\nu}(0)\,|\,p_{1},p_{2}\,\rangle_{\alpha_{1}\alpha_{2}}%
^{in}=\mathbf{C}_{\alpha_{1}\alpha_{2}}\bar{v}(\theta_{1})\gamma^{\mu}%
u(\theta_{2})\,\tfrac{1}{2}(p_{1}^{\nu}-p_{2}^{\nu})\,F_{0}(\theta_{12})
\]
with $F_{0}(\theta)$ given by (\ref{Fmin6}) and (\ref{Fpm06}) which agrees
with the results of \cite{BFK8}.
\end{proposition}

\begin{proof}
The more general proof in Appendix \ref{a2} implies for $\nu=1/2$%
\[
\operatorname*{Res}_{\theta_{12}=\frac{1}{2}i\pi}\operatorname*{Res}%
_{\theta_{34}=\frac{1}{2}i\pi}F_{1234}^{SU(4),\mathcal{O}}(\theta_{1}%
,\ldots,\theta_{4})=const.\,\frac{\Gamma\left(  \frac{3}{4}-\frac{1}{2}%
\frac{\omega}{i\pi}\right)  \Gamma\left(  -\frac{1}{4}+\frac{1}{2}\frac
{\omega}{i\pi}\right)  }{\Gamma\left(  \frac{3}{2}-\frac{1}{2}\frac{\omega
}{i\pi}\right)  \Gamma\left(  \frac{1}{2}+\frac{1}{2}\frac{\omega}{i\pi
}\right)  }F^{O(6)}\left(  \omega\right)  .
\]
with $\omega=\theta_{(12)}-\theta_{(34)}$. Together with (\ref{Fmin6}) and
(\ref{Fpm06}) the claim (\ref{EMN}) follows.
\end{proof}

\subsection{The $\mathbf{O(6)}$ Gross-Neveu field $\mathbf{\psi(x)}$}

\paragraph{The $SU(4)$ form factor:}

We follow \cite{BFK1} and define the $SU(4)$ operator $\mathcal{O}%
^{AB}=\left[  \psi^{A},\psi^{B}\right]  $ where $\psi^{A}(x)$ is the
fundamental field of the chiral $SU(4)$-Gross-Neveu model. It has the weight
vector $w^{\mathcal{O}}=(1,1,0,0)$. We write the highest weight component
$\left[  \psi^{1},\psi^{2}\right]  $ as $\mathcal{O}$ and propose the
p-function (see subsection 4.2 of \cite{BFK1})%
\begin{equation}
p^{\mathcal{O}^{(\pm)}}(\underline{\theta},\underline{z})=\left(
p^{\psi^{(\pm)}}(\underline{\theta},\underline{z})\right)  ^{2}=e^{\pm\left(
\sum_{i=1}^{m}z_{i}-\frac{3}{4}\sum_{i=1}^{n}\theta_{i}\right)  } \label{pf}%
\end{equation}
belonging to the $\pm$ spinor components. The form factors are again given by
(\ref{FK}) and (\ref{K}). The numbers of \textquotedblleft weight
flip\textquotedblright\ operators in the various levels of the nested Bethe
ansatz are given by (\ref{w}) as $n=2+4L,\,n_{1}=1+3L,\,n_{2}=2L,\,n_{3}=L$.

\paragraph{The $O(6)$ form factor:}

The fundamental $O(6)$ field $\psi^{\alpha}$ is fermionic and transforms as
the vector representation with weight vector $w^{\psi}=(1,0,0)$ \cite{BFK8}.
It is given in terms of $\mathcal{O}^{AB}$ by (\ref{inter1}) and (\ref{map})%
\[
\psi^{\alpha}=M_{(RS)}^{\alpha}\Gamma_{AB}^{(RS)}\left[  \psi^{A},\psi
^{B}\right]  \,.
\]
The bound state formula (\ref{Fn}) applied to (\ref{FK}) and (\ref{K}) with
the p-function (\ref{pf}) yields the $O(6)$ form factor for $n/2$ particles.
In particular we consider the case $L=0$, i.e. $n=2,~m=1$%
\begin{align*}
\operatorname*{Res}_{\theta_{12}=i\pi2/3}K_{\underline{A}}^{SU(4),\mathcal{O}%
^{(\pm)}}(\underline{\theta})  &  =\operatorname*{Res}_{\theta_{12}=i\pi
2/3}\int_{\mathcal{C}_{\underline{\theta}}}dz\,\tilde{\phi}\left(  \theta
_{1}-z\right)  \tilde{\phi}\left(  \theta_{2}-z\right)  e^{\pm\left(
z-\tfrac{3}{4}\left(  \theta_{1}+\theta_{2}\right)  \right)  }\,\tilde{\Psi
}_{\underline{A}}(\underline{\theta},z)\\
&  =\tilde{\phi}\left(  \theta_{12}\right)  e^{\pm\left(  \theta_{2}-\tfrac
{3}{4}\left(  \theta_{1}+\theta_{2}\right)  \right)  }\operatorname*{Res}%
_{\theta_{12}=i\pi2/3}\tilde{S}_{A_{1}A_{2}}^{21}\left(  \theta_{12}\right)
\end{align*}
where pinching at $z=\theta_{2}$ was used. Therefore the $O(6)$ one particle
form factor of the field is with $\theta=\frac{1}{2}\left(  \theta_{1}%
+\theta_{2}\right)  $ (up to const.)%
\[
F_{1}^{O(6),\psi^{(\pm)}}(\theta)=e^{\mp\frac{1}{2}\theta}=u^{(\pm)}(\theta)
\]
as expected.

\section{$\mathbf{O(6)\simeq SU(4)}$ as a start of level iteration for
$\mathbf{O(N)}$}

\label{s5}

\subsection{The modified $\mathbf{n}$-particle K-function for $\mathbf{O(6)}$}

The $O(N)$ Gross-Neveu form factors are given by the `off-shell'
\textbf{nested} Bethe ansatz \cite{BFK8}. Therefore we need the higher level
$O(N-2k)$ Bethe ansatz for $k=1,\dots,N/2-3$. The last one is of $O(6)$ type.
For this discussion it is convenient to introduce the variables $u,v$ with
$\theta=i\pi\nu_{k}u,~z=i\pi\nu_{k}v$ with $\nu_{k}=2/(N-2k-2)$. For the
$O(N-2k)$ S-matrix $S^{(k)}(u)$ we write
\begin{align}
\tilde{S}^{(k)}(u)  &  =S^{(k)}/S_{+}^{(k)}=\tilde{b}(u)\mathbf{1}+\tilde
{c}(u)\mathbf{P}+\tilde{d}_{k}(u)\mathbf{K}\label{Su}\\
\tilde{b}(u)  &  =\frac{u}{u-1},~\tilde{c}(u)=\frac{-1}{u-1},~\tilde{d}%
_{k}(u)=\frac{u}{u-1}\frac{1}{u-1/\nu_{k}}\,.\nonumber
\end{align}
and define the higher level K-functions%
\begin{align}
K_{\underline{\alpha}}^{(k)}(\underline{u})  &  =\tilde{N}_{m_{k}}^{(k)}%
\int_{\mathcal{C}_{\underline{u}}^{(1)}}dv_{1}\cdots\int_{\mathcal{C}%
_{\underline{u}}^{(m_{k})}}dv_{m_{k}}\,h(\underline{u},\underline{v}%
)p^{(k)}(\underline{u},\underline{v})\,\,\tilde{\Psi}_{\underline{\alpha}%
}^{(k)}(\underline{u},\underline{v})\label{Kk}\\
\tilde{\Psi}_{\underline{\alpha}}^{(k)}(\underline{u},\underline{v})  &
=K_{\underline{\mathring{\beta}}}^{(k+1)}(\underline{v})\,\big(\tilde{\Phi
}^{(k)}\big)_{\underline{\alpha}}^{\underline{\mathring{\beta}}}(\underline
{u},\underline{v})\nonumber
\end{align}
with $\underline{u}=u_{1},\dots,u_{n_{k}},~\underline{v}=v_{1},\dots,v_{m_{k}%
}$ and $m_{k}=n_{k+1}$. The basic Bethe ansatz co-vectors $\big(\tilde{\Phi
}^{(k)}\big)_{\underline{\alpha}}^{\underline{\mathring{\beta}}}(\underline
{u},\underline{v})$ are defined analogously to (\ref{1.38}). The function
$h(\underline{u},\underline{v})$ is given by (\ref{h}) and (\ref{tau}) where
$\tilde{\phi}\left(  \theta\right)  $ is replaced by
\[
\tilde{\phi}_{\nu}\left(  \theta\right)  =\Gamma\left(  1-\tfrac{1}{2}%
\nu+\tfrac{1}{2\pi i}\theta\right)  \Gamma\left(  -\tfrac{1}{2\pi i}%
\theta\right)  ,~\nu=\nu_{0}=2/(N-2)
\]

The higher level K-functions $K_{\underline{\alpha}}^{(k)}(\underline{u})$ for
$k>0$ satisfy the equations

\begin{itemize}
\item[(i)$^{(k)}$]
\begin{equation}
K_{\dots ij\dots}^{(k)}(\dots,u_{i},u_{j},\dots)=K_{\dots ji\dots}^{(k)}%
(\dots,u_{j},u_{i},\dots)\,\tilde{S}_{ij}^{(k)}(u_{ij}) \label{ik}%
\end{equation}

\item[(ii)$^{(k)}$]
\begin{equation}
K_{1\ldots n_{k}}^{(k)}(u_{1}+2/\nu,u_{2},\dots,u_{n_{k}})\sigma
_{1}^{\mathcal{O}}\mathbf{C}^{\bar{1}1}=K_{2\ldots n_{k}1}^{(k)}(u_{2}%
,\dots,u_{n_{k}},u_{1})\mathbf{C}^{1\bar{1}} \label{iik}%
\end{equation}

\item[(iii)$^{(k)}$]
\begin{equation}
\operatorname*{Res}_{u_{12}=1/\nu_{k}}K_{1\dots n_{k}}^{(k)}(u_{1}%
,\dots,u_{n_{k}})=\prod_{i=3}^{n_{k}}\tilde{\phi}_{\nu}(u_{i1}+1)\tilde{\phi
}_{\nu}(u_{i2})\mathbf{C}_{12}K_{3\dots n_{k}}^{(k)}(u_{3},\dots,u_{n_{k}})\,.
\label{iiik}%
\end{equation}

\end{itemize}

The normal form factor equations (i) - (iii) for $O(N-2k)$ are similar to
these higher level equations. There are, however, two differences:

\begin{enumerate}
\item The shift in (ii)$^{(k)}$ is the one of $O(N)$ but not that of $O(N-2k)$.

\item There is only one term on the right hand side in (iii)$^{(k)}$.
\end{enumerate}

In particular for $k=N/2-3=1/\nu-2$ we have $\nu_{k}=\frac{1}{2}$ and
$K_{\underline{\alpha}}^{(k)}(\underline{u})=K_{\underline{\alpha}}^{O(6),\nu
}(\underline{u})$ is of $O(N-2k)=O(6)$ type, which means in particular that
the S-matrix and the Bethe state are the ones of $O(6)$. We call
$K_{\underline{\alpha}}^{O(6),\nu}$ \textbf{a modified $\mathbf{O(6)}$
K-function}.

\subsection{The modified $\mathbf{n}$-particle K-function for $\mathbf{SU(4)}$}

Replacing in (\ref{K}) and (\ref{phi4})
\[
\tilde{\phi}\left(  \theta\right)  \rightarrow\tilde{\phi}_{\nu}\left(
\theta\right)  =\Gamma\left(  1-\tfrac{1}{2}\nu+\tfrac{1}{2\pi i}%
\theta\right)  \Gamma\left(  -\tfrac{1}{2\pi i}\theta\right)
\]
we obtain the \textbf{modified }$\mathbf{n}$\textbf{-particle K-function for
}$\mathbf{SU(4)}$ which satisfies the form factor equation (ii) (see
(\ref{1.12})) not for the shift $\theta_{1}\rightarrow\theta_{1}+2\pi i$ but
for $\theta_{1}\rightarrow\theta_{1}+i\pi/\nu$ and in (iii) (see (\ref{1.14}))
the second term on the right hand side is missing. Again we use for
convenience the variables $u$ and $v$ with $\theta=i\pi\nu u,~z=i\pi\nu v$,
then the K-function
\begin{equation}
K_{\underline{A}}^{SU(4),\nu}(\underline{u},\underline{\nu})=\int
_{\mathcal{C}_{\underline{u},\nu}}d\underline{v}\prod_{i=1}^{n}\prod_{j=1}%
^{m}\tilde{\phi}_{\nu}(u_{i}-v_{j})\prod_{i<j}\tau_{\nu}(v_{ij})p(\underline
{u},\underline{v})\tilde{\Psi}_{\underline{A}}(\underline{u},\underline{v})
\label{K4nu}%
\end{equation}
satisfies for $\nu<\frac{1}{2}$ not the form factor equations (ii) and (iii)
of (\ref{1.10}) - (\ref{1.18}) but the modified ones

\begin{itemize}
\item[(ii)$_{\nu}$]
\begin{equation}
K_{1\ldots n}^{SU(4),\nu}(u_{1}+2/\nu,u_{2},\dots,u_{n})\sigma_{1}%
^{\mathcal{O}}\mathbf{C}^{\bar{1}1}=K_{2\ldots n1}^{SU(4),\nu}(u_{2}%
,\dots,u_{n},u_{1})\mathbf{C}^{1\bar{1}} \label{iinu}%
\end{equation}

\item[(iii)$_{\nu}$]
\begin{equation}
\operatorname*{Res}_{u_{34}=1}\operatorname*{Res}_{u_{23}=1}%
\operatorname*{Res}_{u_{12}=1}K_{1\ldots n}^{SU(4),\nu}(\underline{u}%
)=\prod_{i=5}^{n}\prod_{j=2}^{4}\tilde{\phi}_{\nu}(u_{ij})\,\varepsilon
_{1234}K_{5\dots n}^{SU(4),\nu}(\underline{\check{u}}) \label{iiinu}%
\end{equation}
with $\underline{\check{u}}=u_{5},\dots,u_{n}$.
\end{itemize}

The proofs of these equations are quite analogous to the ones in \cite{BFK1}
for the normal $SU(N)$ K-functions for $N=4$.

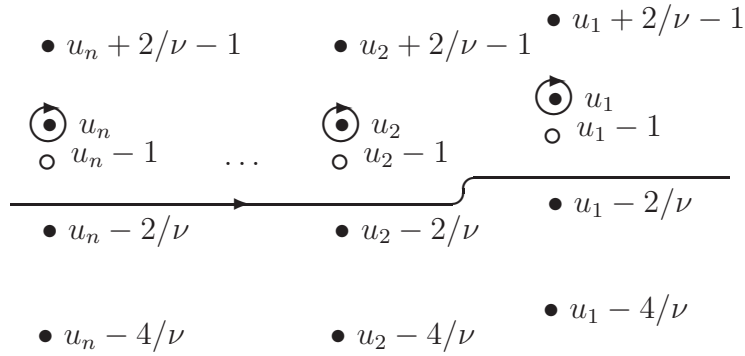
\begin{figure}[tbh]%
\[
\unitlength3.5mm\begin{picture}(26,15)(0.7,-3) \thicklines
\def\ff#1{
\put(.1,9.2){${\bullet~u_{#1}+2/\nu-1}$}
\put(0.,-1.8){${\bullet~u_{#1}-4/\nu}$}
\put(0.16,2.2){${\bullet~u_{#1}-2/\nu}$}
\put(.36,5.1){${\hbox{\circle{.5}}}~{u_{#1}-1}$}
\put(.16,6.2){${\bullet~~u_{#1}}$}
\put(.4,6.55){\hbox{\circle{1.3}}}
\put(.85,7.15){\vector(1,0){0}}
}
\put(1,0){\ff{n}}
\put(8,5){\dots}
\put(12,0){\ff{2}}
\put(20,1){\ff{1}}
\put(9,3.5){\vector(1,0){0}}
\put(0,4){\oval(34,1)[br]}
\put(27,4){\oval(20,1)[tl]}
\end{picture}
\]
\caption{The integration contour $\mathcal{C}_{\underline{u},\nu}$ in
(\ref{K4nu}). The bullets refer to poles of the integrand resulting from
$\,\tilde{\phi}(u_{i}-v_{j})$ and the small open circles refer to poles
originating from $\tilde{b}(u_{i}-v_{j})$ and $\tilde{c}(u_{i}-v_{j})$.}%
\end{figure}

\subsection{$\mathbf{n}^{\prime}\mathbf{=n/2}$ bound states of $\mathbf{SU(4)}%
$ particles:}

We apply the bound state formula (iv) to an $n$-particle modified K-function
of $SU(4)$ and define correspondingly to (\ref{Kn}) for $\nu=2/(N-2)$ an
$n^{\prime}=n/2$-particle $O(6)$ K-function
\begin{equation}
K_{\underline{\alpha}}^{O(6),\nu}(\underline{o})\Gamma_{\underline{A}%
}^{\underline{\alpha}}=\prod\limits_{1\leq i<j\leq n^{\prime}}\frac{1}%
{\tilde{\phi}_{\nu}(-o_{ij})\tilde{\phi}_{\nu}(-o_{ij}+1)}\operatorname*{Res}%
_{u_{12}=1}\ldots\operatorname*{Res}_{u_{n-1n}=1}K_{\underline{A}}^{SU(4),\nu
}(\underline{u}) \label{Knu}%
\end{equation}
with $o_{i}=\frac{1}{2}\left(  u_{2i-1}+u_{2i}\right)  $ and the intertwiner
$\Gamma_{\underline{A}}^{\underline{\alpha}}=\Gamma_{A_{1}A_{2}}^{\alpha_{1}%
}\ldots\Gamma_{A_{n-1}A_{n}}^{\alpha_{n^{\prime}}}$. Correspondingly to lemma
\ref{l1} we prove

\begin{lemma}
\label{l2}The K-function defined by (\ref{Knu}) satisfies the modified form
factor equations

\begin{itemize}
\item[$(i)_{\nu}$]
\[
K_{\dots ij\dots}^{O(6),\nu}(\dots o_{i},o_{j}\dots)=K_{\dots ji\dots
}^{O(6),\nu}(\dots o_{j},o_{i}\dots)\tilde{S}^{O(6)}(o_{ij})
\]

\item[$(ii)_{\nu}$]
\[
K_{12\dots n^{\prime}}^{O(6),\nu}(o_{1}+2/\nu,o_{2},\dots,o_{n^{\prime}%
})\mathbf{C}^{\bar{1}1}=K_{2\dots n^{\prime}1}^{O(6),\nu}(o_{2},\dots
,o_{n^{\prime}},o_{1})\mathbf{C}^{1\bar{1}}%
\]

\item[$(iii)_{\nu}$]
\[
\operatorname*{Res}_{o_{12}=2}K_{1\dots n^{\prime}}^{O(6),\nu}(\underline
{o})=\prod_{i=3}^{n^{\prime}}\tilde{\phi}_{\nu}(o_{i1}+1)\tilde{\phi}_{\nu
}(o_{i2})\mathbf{C}_{12}K_{3\dots n^{\prime}}^{O(6),\nu}(\underline{\check{o}%
})
\]
with $\underline{\check{o}}=o_{3},\ldots o_{n^{\prime}}$.
\end{itemize}
\end{lemma}

\begin{proof}
We follow here the proof of Proposition 7 in \cite{BK}. For $(i)_{\nu}$ and
$(ii)_{\nu}$ the proofs are again obvious. To prove $(iii)_{\nu}$ one follows
Appendix E of \cite{BK} taking into account that also in (\ref{iiinu}) there
is only one term on the right hand side.
\end{proof}

\begin{corollary}
The K-function defined by (\ref{Knu}) satisfies the higher level equations
$(i)^{(k)}$ - $(iv)^{(k)}$ or (4.13) - (4.16) of \cite{BFK8} for $k=N/2-3$,
i.e. $\nu_{k}=1/2$. Therefore it serves as a starting of the nesting for the
construction of an $O(N)$-Gross-Neveu form factor for arbitrary even $N>6$.
\end{corollary}

To construct the form factors of the $O(N)$ Gross-Neveu model for the
operators $\bar{\psi}\psi,\,J_{\mu}^{\alpha\beta}$ and $\psi^{\alpha}$ with
weight vectors $w=\left(  0,0,\dots,0\right)  ,\,\left(  1,0,\dots,0\right)  $
and $\left(  1,1,0,\dots,0\right)  $, respectively, we need for the starting
of the nested Bethe ansatz the modified $O(6)$ one for the iso-scalar with
weight vectors $w=\left(  0,0,0\right)  $. Therefore we generalize the
constructions of Subsection \ref{sscalar} from $\nu=1/2$ to general $\nu$ and prove

\begin{lemma}
\label{L3}The bound state formula (\ref{Knu}) applied to the modified $SU(4)$
K-function (\ref{K4nu}) with the p-function (\ref{ps})
\begin{equation}
p\left(  \underline{u},\underline{v},\underline{w},\underline{x}\right)
=e^{i\pi\nu\left(  \frac{1}{2}\sum_{i=1}^{4L}u_{i}-\sum_{i=1}^{3L}v_{i}%
+\sum_{i=1}^{L}x_{i}\right)  }-1 \label{p4s}%
\end{equation}
for $L=1,2,\dots$ (see (\ref{w})) yields the modified $O(6)$ K-function for
the iso-scalar for $n^{\prime}=2L$ particles. This means that for $L=1$ the
the bound state formula (\ref{Knu}) yields the modified $O(6)$ two-particle
K-function
\begin{equation}
K_{\alpha_{1}\alpha_{2}}(o_{1},o_{2})=\mathbf{C}_{\alpha_{1}\alpha_{2}}%
\frac{\Gamma\left(  1-\frac{1}{2}\nu-\frac{1}{2}\nu o_{12}\right)
\Gamma\left(  -\frac{1}{2}\nu+\frac{1}{2}\nu o_{12}\right)  }{\Gamma\left(
1+\nu-\frac{1}{2}\nu o_{12}\right)  \Gamma\left(  \nu+\frac{1}{2}\nu
o_{12}\right)  }\,. \label{k0}%
\end{equation}
This is the higher level K-function needed as the starting for the nested
$O(N)$ Bethe ansatz (see \cite{BFK8}).
\end{lemma}

The proof of this lemma can be found in Appendix \ref{a2}. It follows the main
result of this article:

\begin{corollary}
For all $O(N)$ Gross-Neveu form factors of operators $\mathcal{O}(x)$ with
weights $w^{\mathcal{O}}=(w_{1},w_{2},0,\dots,0,0)$ the start of the nesting
is obtained by (\ref{K4nu}) with the p-function (\ref{p4s}) and the bound
state formula (\ref{Knu}).
\end{corollary}

\subsection*{Conclusions:}

\addcontentsline{toc}{section}{Conclusions}

The form factors for the $SU(N)$ chiral Gross-Neveu model were constructed in
\cite{BFK0,BFK1,BFK2,BFK3,BFK4}. In \cite{BFK7} we used the isomorphism
$O(4)\simeq SU(2)\times SU(2)$ as the starting point of the nesting procedure
to construct the $O(N)~\sigma$-model form factors. Up to now we were not able
to do the analog for the $O(N)$ Gross-Neveu model. However, the fundamental
particles of the $O(6)$ Gross-Neveu model may by identified with the bound
states of the $SU(4)$ chiral Gross-Neveu model \cite{KT1}. Using this
identification we showed in the present article how to use the $O(6)$
functions as the starting point of the nesting procedure to construct the
$O(N)$ Gross-Neveu model form factors (for $N$ even). In a forthcoming article
we will consider the $O(4)$ Gross-Neveu model. Also the asymptotic behavior of
the form factors and the short distance behavior of the correlation functions
will be investigated.

\paragraph{Acknowledgment:}

The authors have profited from discussions with R. Schra\-der and B. Schroer.
H.~B. was supported by Armenian grant 15T-1C308 and by ICTP OEA-AC-100
project. A.~F. acknowledges financial support from CNPq (Conselho Nacional de
Desenvolvimento Cientifico e Tecnologico). M.~K. was supported by Fachbereich
Physik, Freie Universit\"{a}t Berlin grant 01000/20000000.

\appendix

\section*{Appendix}

\addcontentsline{toc}{part}{Appendix}

\renewcommand{\theequation}{\mbox{\Alph{section}.\arabic{equation}}} \setcounter{equation}{0}

For simplicity the equations in the following are mostly written up to
inessential constants.

\section{Proof of Lemma {\ref{L3}}}

\label{a2}\label{here}\label{1}\label{2}

\begin{proof}
We have $n=4,\,n_{1}=3,\,n_{2}=2$ and $n_{3}=1$. For convenience we use again
the variables $u,v,w,x,o$ with $\theta=i\pi\nu u,~z^{(1)}=i\pi\nu
v,~z^{(2)}=i\pi\nu w,~z^{(3)}=i\pi\nu x,~\omega=i\pi\nu o$. We prove that
(\ref{Knu}) for $n=4$ (with $o_{1}=(u_{1}+u_{2})/2,~o_{2}=(u_{3}+u_{4})/2$)%
\[
K_{\underline{\alpha}}^{O(6),\nu}(\underline{o})\Gamma_{\underline{A}%
}^{\underline{\alpha}}=\frac{1}{\tilde{\phi}_{\nu}(-o_{12})\tilde{\phi}_{\nu
}(-o_{12}+1)}\operatorname*{Res}_{u_{12}=1}\operatorname*{Res}_{u_{34}%
=1}K_{\underline{A}}^{SU(4),\nu}(\underline{u})
\]
with the p-function (\ref{p4s}) implies (\ref{k0})
\[
K_{\alpha_{1}\alpha_{2}\alpha}^{O(6),\nu}(o_{1},o_{2})=\mathbf{C}_{\alpha
_{1}\alpha_{2}}\frac{\Gamma\left(  1-\frac{1}{2}\nu-\frac{1}{2}\nu
o_{12}\right)  \Gamma\left(  -\frac{1}{2}\nu+\frac{1}{2}\nu o_{12}\right)
}{\Gamma\left(  1+\nu-\frac{1}{2}\nu o_{12}\right)  \Gamma\left(  \nu+\frac
{1}{2}\nu o_{12}\right)  }\,.
\]
We calculate the residues (first for $p=1$) of the component with
$\underline{A}=(1,2,3,4)$ (using pinching at $v_{1}=u_{2}\rightarrow
u_{1}-1,~v_{2}=u_{4}\rightarrow u_{3}-1$)%
\begin{align}
X  &  =\operatorname*{Res}_{u_{12}=1}\operatorname*{Res}_{u_{34}=1}%
K_{1234}^{SU(4)}(\underline{u})=\operatorname*{Res}_{u_{12}=1}%
\operatorname*{Res}_{u_{34}=1}\int_{\mathcal{C}_{\underline{u}}}%
dv_{3}\operatorname*{Res}_{u_{12}=1}\operatorname*{Res}_{u_{34}=1}%
\operatorname*{Res}_{v_{1}=u_{2}}\operatorname*{Res}_{v_{2}=u_{4}}h\left(
\underline{u},\underline{v}\right)  \Psi_{1234}(\underline{u},\underline
{v})\label{X}\\
&  =\int_{\mathcal{C}_{\underline{u}}}dv_{3}\operatorname*{Res}_{u_{12}%
=1}\operatorname*{Res}_{u_{34}=1}\operatorname*{Res}_{v_{1}=u_{2}%
}\operatorname*{Res}_{v_{2}=u_{4}}h\left(  \underline{u},\underline{v}\right)
K_{\underline{B}}^{(1)}(\underline{v})\Phi_{1234}^{\underline{B}}%
(\underline{u},\underline{v})\nonumber\\
&  =\left[  \tilde{b}\left(  u_{14}\right)  \int_{\mathcal{C}_{\underline{u}}%
}dv_{3}h_{r}\left(  \underline{u},\underline{v}\right)  \left(  K_{234}%
^{(1)}(\underline{v})-K_{243}^{(1)}(\underline{v})\right)  \tilde{b}\left(
u_{1}-v_{3}\right)  \tilde{c}\left(  u_{3}-v_{3}\right)  \right]
_{\substack{v_{1}=u_{2},v_{2}=u_{4}~~~~~~\\u_{1}=u_{2}+1,u_{3}=u_{4}%
+1}}\nonumber
\end{align}
with $h_{r}\left(  \underline{u},\underline{v}\right)  =\operatorname*{Res}%
\limits_{v_{1}=u_{2}}\operatorname*{Res}\limits_{v_{2}=u_{4}}h\left(
\underline{u},\underline{v}\right)  $. It was used that for $v_{1}=u_{2}$ and
$v_{2}=u_{4}$ (see Fig. \ref{f0})\begin{figure}[th]
\begin{center}
$%
\begin{array}
[c]{c}%
\unitlength1mm\begin{picture}(53,50)(38,2)\scriptsize \put(87,4){\line(0,1){17}} \put(82,4){\line(0,1){17}} \put(77,4){\line(0,1){17}} \put(72,4){\line(0,1){17}} \put(90,9){\line(-1,0){22}} \put(90,13){\line(-1,0){22}} \put(90,17){\line(-1,0){22}} \put(68,37){\oval(7,40)[lb]} \put(68,37){\oval(16,48)[lb]} \put(68,37){\oval(25,56)[lb]} \put(68,46){\oval(39,26)[lb]} \put(68,46){\oval(48,35)[lb]} \put(50,50){\oval(22,16)[lb]} \put(68.8,5){$u_1$} \put(73.6,5){$u_2$} \put(78.8,5){$u_3$} \put(83.8,5){$u_4$} \put(52.8,25){$v_1$} \put(57.,25){$v_2$} \put(61.5,25){$v_3$} \put(40.1,38){$w_1$} \put(44.7,38){$w_2$} \put(36.7,48){$x$} \put(91,8){1} \put(91,12){1} \put(91,16){1} \put(87,22){1} \put(82,22){1} \put(77,22){1} \put(72,22){1} \put(87,1){4} \put(82,1){3} \put(77,1){2} \put(72,1){1} \put(69,28){2} \put(69,32){2} \put(55,38){2} \put(60,38){2} \put(64,38){2} \put(51,41){3} \put(49,47){3} \put(44,47){3} \put(39,51){4} \end{picture}
\end{array}
\rightarrow%
\begin{array}
[c]{c}%
\unitlength1mm\begin{picture}(53,50)(38,2)\scriptsize \put(82,4){\line(0,1){17}} \put(72,4){\line(0,1){17}} \put(68,4){\oval(18,10)[rt]} \put(68,4){\oval(38,18)[rt]} \put(90,17){\line(-1,0){22}} \put(68,37){\oval(7,40)[lb]} \put(68,37){\oval(16,48)[lb]} \put(68,37){\oval(25,56)[lb]} \put(68,46){\oval(39,26)[lb]} \put(68,46){\oval(48,35)[lb]} \put(50,50){\oval(22,16)[lb]} \put(68.8,5){$u_1$} \put(73.5,5){$u_2$} \put(78.8,5){$u_3$} \put(83.8,5){$u_4$} \put(52.,25){$u_2$} \put(56.6,25){$u_4$} \put(61.5,25){$v_3$} \put(40.1,38){$w_1$} \put(44.7,38){$w_2$} \put(36.7,48){$x$} \put(91,16){1} \put(82,22){1} \put(72,22){1} \put(73,14){1} \put(73,10){1} \put(63,6){2} \put(67,18){3,4} \put(66,13.7){4,3} \put(77,18){3,4} \put(77,13.7){4,3} \put(87,1){4} \put(82,1){3} \put(77,1){2} \put(72,1){1} \put(69,28){2} \put(69,32){2} \put(55,38){2} \put(60,38){2} \put(64,38){2} \put(51,41){3} \put(49,47){3} \put(44,47){3} \put(39,51){4} \end{picture}
\end{array}
$
\end{center}
\caption{The Bethe state $\Psi_{\underline{A}}(\underline{u},\underline{v})$
in (\ref{K}) for an iso-scalar operator where $\underline{A}=(1,2,3,4)$ and
$v_{1}\rightarrow u_{2},v_{2}\rightarrow u_{4}$.}%
\label{f0}%
\end{figure}
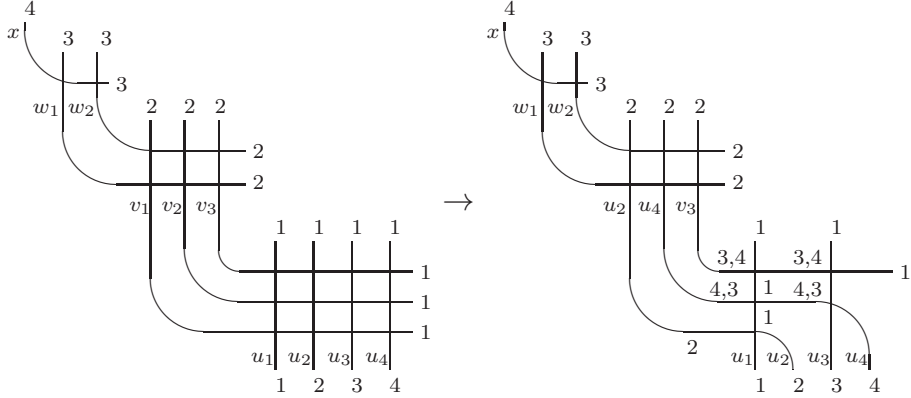$%
%TCIMACRO{\FRAME{itbpFw}{4.3502in}{1.9871in}{0in}{}{}{fstate.eps}%
%{\special{ language "Scientific Word";  type "GRAPHIC";
%maintain-aspect-ratio TRUE;  display "USEDEF";  valid_file "F";
%width 4.3502in;  height 1.9871in;  depth 0in;  original-width 5.3718in;
%original-height 2.4379in;  cropleft "0";  croptop "1";  cropright "1";
%cropbottom "0";  filename 'graphics/fstate.eps';file-properties "XNPEU";}}}%
%BeginExpansion
{\phantom{\rule{4.3502in}{1.9871in}}}%
%EndExpansion
$%
\[
\operatorname*{Res}_{u_{12}=1}\operatorname*{Res}_{u_{34}=1}\Phi
_{1234}^{\underline{B}}(\underline{u},u_{2},u_{4})=\delta_{2}^{B_{1}}\left(
\delta_{3}^{B_{2}}\delta_{4}^{B_{3}}-\delta_{4}^{B_{2}}\delta_{3}^{B_{3}%
}\right)  \tilde{b}\left(  u_{1}-u_{4}\right)  \tilde{b}\left(  u_{1}%
-v_{3}\right)  \tilde{c}\left(  u_{3}-v_{3}\right)
\]
and therefore%
\begin{align*}
&  K_{\underline{B}}^{(1)}(u_{2},u_{4},v_{3})\operatorname*{Res}_{u_{12}%
=1}\operatorname*{Res}_{u_{34}=1}\Phi_{1234}^{\underline{B}}(\underline
{u},u_{2},u_{4},v_{3})\\
&  =\left(  K_{234}^{(1)}(u_{2},u_{4},v_{3})-K_{243}^{(1)}(u_{2},u_{4}%
,v_{3})\right)  \tilde{b}\left(  u_{1}-u_{4}\right)  \tilde{b}\left(
u_{1}-v_{3}\right)  \tilde{c}\left(  u_{3}-v_{3}\right)  \,,
\end{align*}
further with $o=o_{12}=u_{24}=v_{12}$%
\begin{multline*}
X(o)=\operatorname*{Res}_{u_{12}=1}\operatorname*{Res}_{u_{34}=1}%
K_{1234}^{SU(4)}(\underline{u})\\
=\int_{\mathcal{C}_{\underline{u}}}dv_{3}\prod_{i=1}^{4}\prod_{\underset
{i,j\neq2,1;4,2}{j=1}}^{2}\tilde{\phi}_{\nu}(u_{i}-v_{j})\tau(v_{12}%
)\prod_{i=1}^{4}\tilde{\phi}_{\nu}(u_{i}-v_{3})\tau(v_{13})\tau(v_{23})\\
\times\left(  K_{234}^{(1)}(u_{2},u_{4},v_{3})-K_{243}^{(1)}(u_{2},u_{4}%
,v_{3})\right)  \tilde{b}\left(  u_{1}-u_{4}\right)  \tilde{b}\left(
u_{1}-v_{3}\right)  \tilde{c}\left(  u_{3}-v_{3}\right)
\end{multline*}
with $v_{1}\rightarrow u_{2},~v_{2}\rightarrow u_{4},~u_{1}\rightarrow
u_{2}+1,~u_{3}\rightarrow u_{4}+1$. We get $X$ as%
\begin{align*}
X(o)  &  =\frac{\Gamma\left(  1-\tfrac{1}{2}\nu\left(  1+o\right)  \right)
\Gamma\left(  \tfrac{1}{2}\nu\left(  o-1\right)  \right)  }{\sin\frac{1}{2}%
\pi\nu o}Y(o)\\
Y(o)  &  =\int_{\mathcal{C}_{\underline{u}}}dv_{3}\tilde{c}\left(
-u_{4}+v_{3}\right)  \left(  K_{234}^{(1)}(u_{2},u_{4},v_{3})-K_{243}%
^{(1)}(u_{2},u_{4},v_{3})\right)
\end{align*}
where it was used that for $v_{1}=u_{2},~v_{2}=u_{4},~u_{1}=u_{2}%
+1,~u_{3}=u_{4}+1$%
\begin{align*}
&  \tilde{b}\left(  u_{1}-u_{4}\right)  \frac{\tilde{\phi}_{\nu}(u_{1}%
-v_{1})\tilde{\phi}_{\nu}(u_{1}-v_{2})\tilde{\phi}_{\nu}(u_{2}-v_{2}%
)\tilde{\phi}_{\nu}(u_{3}-v_{1})\tilde{\phi}_{\nu}(u_{3}-v_{2})\tilde{\phi
}_{\nu}(u_{4}-v_{1})}{\tilde{\phi}_{\nu}(v_{1}-v_{2})\tilde{\phi}_{\nu}%
(-v_{1}+v_{2})}\\
&  =\frac{1}{\sin\tfrac{1}{2}\nu\pi\left(  u_{4}-u_{2}\right)  }\Gamma\left(
1+\tfrac{1}{2}\nu\left(  -u_{2}-1+u_{4}\right)  \right)  \Gamma\left(
-\tfrac{1}{2}\nu\left(  u_{4}+1-u_{2}\right)  \right)
\end{align*}
and
\[
\frac{\tilde{\phi}_{\nu}(u_{1}-v_{3})\tilde{\phi}_{\nu}(u_{2}-v_{3}%
)\tilde{\phi}_{\nu}(u_{3}-v_{3})\tilde{\phi}_{\nu}(u_{4}-v_{3})}{\tilde{\phi
}_{\nu}(v_{1}-v_{3})\tilde{\phi}_{\nu}(-v_{1}+v_{3})\tilde{\phi}_{\nu}%
(v_{2}-v_{3})\tilde{\phi}_{\nu}(-v_{2}+v_{3})}\frac{\tilde{b}\left(
u_{1}-v_{3}\right)  \tilde{c}\left(  u_{3}-v_{3}\right)  }{\tilde{c}\left(
-u_{4}+v_{3}\right)  }=-1\,.
\]
$\allowbreak$Therefore we have%
\[
K_{\underline{\alpha}}^{O(6),\nu}(\underline{o})=\mathbf{C}_{\alpha_{1}%
\alpha_{2}}\frac{\Gamma\left(  1-\tfrac{1}{2}\nu\left(  1+o\right)  \right)
\Gamma\left(  \tfrac{1}{2}\nu\left(  o-1\right)  \right)  }{\tilde{\phi}_{\nu
}(-o)\tilde{\phi}_{\nu}(-o+1)\sin\frac{1}{2}\pi\nu o}Y(o)=\mathbf{C}%
_{\alpha_{1}\alpha_{2}}Y(o)\,.
\]
Next we calculate%
\[
K_{\underline{B}}^{(1)}(\underline{v})=\int_{\mathcal{C}_{\underline{v}}%
}d\underline{w}h\left(  \underline{v},\underline{w}\right)  K_{\underline{C}%
}^{(2)}(\underline{w})\Phi^{(1)}\,_{\underline{B}}^{\underline{C}}%
(\underline{v},\underline{w})
\]
with (see Fig. \ref{f0})%
\[
\Phi^{(1)}\,_{2B_{2}B_{3}}^{C_{1}C_{2}}(\underline{v},\underline{w}%
)=\delta_{B_{2}}^{C_{1}}\delta_{B_{3}}^{C_{2}}\Phi_{1}+\delta_{B_{2}}^{C_{2}%
}\delta_{B_{3}}^{C_{1}}\Phi_{2}%
\]%
\begin{align*}
\Phi_{2}  &  =\tilde{b}\left(  v_{1}-w_{1}\right)  \tilde{b}\left(
v_{1}-w_{2}\right)  \tilde{b}\left(  v_{2}-w_{1}\right)  \tilde{c}\left(
v_{2}-w_{2}\right)  \tilde{c}\left(  v_{3}-w_{1}\right) \\
\Phi_{1}  &  =\tilde{b}\left(  v_{1}-w_{1}\right)  \tilde{b}\left(
v_{1}-w_{2}\right)  \tilde{c}\left(  v_{2}-w_{1}\right) \\
&  \times\left(  \tilde{b}\left(  v_{2}-w_{2}\right)  \tilde{b}\left(
v_{3}-w_{1}\right)  \tilde{c}\left(  v_{3}-w_{2}\right)  +\tilde{c}\left(
v_{2}-w_{2}\right)  \tilde{c}\left(  v_{3}-w_{1}\right)  \right)
\end{align*}
and%
\begin{align}
K_{\underline{C}}^{(2)}(\underline{w})  &  =\int_{\mathcal{C}_{\underline{w}}%
}dx\tilde{\phi}_{\nu}(w_{1}-x)\tilde{\phi}_{\nu}(w_{2}-x)\left(
\delta_{\underline{C}}^{34}\tilde{b}(w_{1}-x)\tilde{c}(w_{2}-x)+\delta
_{\underline{C}}^{43}\tilde{c}(w_{1}-x)\right) \label{x}\\
&  =\left(  \delta_{\underline{C}}^{34}-\delta_{\underline{C}}^{43}\right)
\Gamma\left(  -\tfrac{1}{2}\nu+\tfrac{1}{2}\nu w_{12}\right)  \Gamma\left(
1-\tfrac{1}{2}\nu-\tfrac{1}{2}\nu w_{12}\right) \nonumber
\end{align}
which follows from%
\begin{multline*}
\frac{1}{2\pi i}\left(  \int_{\mathcal{C}_{a}}+\int_{\mathcal{C}_{b}}\right)
dz\Gamma(a-z)\Gamma(b-z)\Gamma\left(  c+z\right)  \Gamma\left(  d+z\right) \\
=-\frac{\Gamma\left(  c+a\right)  \Gamma\left(  d+a\right)  \Gamma\left(
c+b\right)  \Gamma\left(  d+b\right)  }{\Gamma\left(  c+d+a+b\right)  }\,.
\end{multline*}
Therefore%
\begin{multline*}
K_{2B_{2}B_{3}}^{(1)}(\underline{v})=\int_{\mathcal{C}_{\underline{v}}%
}d\underline{w}h\left(  \underline{v},\underline{w}\right)  \Gamma\left(
-\tfrac{1}{2}\nu+\tfrac{1}{2}\nu w_{12}\right)  \Gamma\left(  1-\tfrac{1}%
{2}\nu-\tfrac{1}{2}\nu w_{12}\right) \\
\times\left(  \delta_{B_{2}}^{3}\delta_{B_{3}}^{4}-\delta_{B_{2}}^{4}%
\delta_{B_{3}}^{3}\right)  \left(  \Phi_{1}-\Phi_{2}\right)
\end{multline*}
because $\left(  \delta_{\underline{C}}^{34}-\delta_{\underline{C}}%
^{43}\right)  \left(  \delta_{B_{2}}^{C_{1}}\delta_{B_{3}}^{C_{2}}\Phi
_{1}+\delta_{B_{2}}^{C_{2}}\delta_{B_{3}}^{C_{1}}\Phi_{2}\right)  =\left(
\delta_{B_{2}}^{3}\delta_{B_{3}}^{4}-\delta_{B_{2}}^{4}\delta_{B_{3}}%
^{3}\right)  \left(  \Phi_{1}-\Phi_{2}\right)  $ and%
\[
\Phi_{1}-\Phi_{2}=\frac{\tilde{b}\left(  v_{1}-w_{1}\right)  \tilde{b}\left(
v_{1}-w_{2}\right)  \tilde{c}\left(  v_{2}-w_{1}\right)  \tilde{c}\left(
v_{2}-w_{2}\right)  }{\tilde{c}\left(  w_{1}-w_{2}\right)  }\frac{\tilde
{c}\left(  v_{3}-w_{1}\right)  \tilde{c}\left(  v_{3}-w_{2}\right)  }%
{\tilde{c}\left(  v_{3}-v_{2}\right)  }\,.
\]
Finally exchanging the integrations%
\begin{align*}
Y(o)  &  =\int_{\mathcal{C}_{\underline{u}}}dv_{3}\tilde{c}\left(
-u_{4}+v_{3}\right)  \left(  K_{234}^{(1)}(u_{2},u_{4},v_{3})-K_{243}%
^{(1)}(u_{2},u_{4},v_{3})\right) \\
&  =\int_{\mathcal{C}_{\underline{v}}}d\underline{w}\tau_{\nu}(w_{12}%
)\Gamma\left(  -\tfrac{1}{2}\nu+\tfrac{1}{2}\nu w\right)  \Gamma\left(
-\tfrac{1}{2}\nu w\right) \\
&  \times\left[  \prod_{i=1}^{2}\prod_{j=1}^{2}\tilde{\phi}_{\nu}(v_{i}%
-w_{j})\frac{\tilde{b}\left(  v_{1}-w_{1}\right)  \tilde{b}\left(  v_{1}%
-w_{2}\right)  \tilde{c}\left(  v_{2}-w_{1}\right)  \tilde{c}\left(
v_{2}-w_{2}\right)  }{\tilde{c}\left(  w_{1}-w_{2}\right)  }\right]
_{v_{1}=u_{2},v_{2}=u_{4}}\\
&  \times\int_{\mathcal{C}_{\underline{w}}}dv_{3}\tilde{c}\left(  -u_{4}%
+v_{3}\right)  \tilde{\phi}_{\nu}(v_{3}-w_{1})\tilde{\phi}_{\nu}(v_{3}%
-w_{2})\frac{\tilde{c}\left(  v_{3}-w_{1}\right)  \tilde{c}\left(  v_{3}%
-w_{2}\right)  }{\tilde{c}\left(  v_{3}-u_{4}\right)  }%
\end{align*}
the $v_{3}$-integration can be done as above in (\ref{x})%
\begin{equation}
\int_{\mathcal{C}_{\underline{w}}}dv_{3}\tilde{\phi}_{\nu}(v_{3}-w_{1}%
)\tilde{\phi}_{\nu}(v_{3}-w_{2})\tilde{c}\left(  v_{3}-w_{1}\right)  \tilde
{c}\left(  v_{3}-w_{2}\right)  =\Gamma(-\tfrac{1}{2}\nu+\tfrac{1}{2}\nu
w_{12})\Gamma(-\tfrac{1}{2}\nu-\tfrac{1}{2}\nu w_{12}) \label{v3}%
\end{equation}
and therefore (for $v_{1}=u_{2},v_{2}=u_{4},~o=u_{2}-u_{4}$)%
\begin{multline*}
Y(o)=\int_{\mathcal{C}_{\underline{v}}}d\underline{w}\left(  \prod_{i=1}%
^{2}\prod_{j=1}^{2}\tilde{\phi}_{\nu}(v_{i}-w_{j})\right) \\
\times\tilde{b}\left(  v_{1}-w_{1}\right)  \tilde{b}\left(  v_{1}%
-w_{2}\right)  \tilde{c}\left(  v_{2}-w_{1}\right)  \tilde{c}\left(
v_{2}-w_{2}\right)  \Psi\left(  w_{1}-w_{2}\right)
\end{multline*}
with%
\begin{align*}
\Psi\left(  w\right)   &  =\frac{\Gamma\left(  -\tfrac{1}{2}\nu+\tfrac{1}%
{2}\nu w\right)  \Gamma\left(  1-\tfrac{1}{2}\nu-\tfrac{1}{2}\nu w\right)
\Gamma\left(  -\tfrac{1}{2}\nu+\tfrac{1}{2}\nu w\right)  \Gamma\left(
-\tfrac{1}{2}\nu-\tfrac{1}{2}\nu w\right)  }{\tilde{c}\left(  w\right)
\tilde{\phi}_{\nu}(w)\tilde{\phi}_{\nu}(-w)}\\
&  =\frac{1}{\pi}w\left(  \sin\tfrac{1}{2}\pi\nu w\right)  \Gamma\left(
-\tfrac{1}{2}\nu+\tfrac{1}{2}\nu w\right)  \Gamma\left(  -\tfrac{1}{2}%
\nu-\tfrac{1}{2}\nu w\right)  \,.
\end{align*}
$\allowbreak$ In (C.10) of \cite{BFK8} was shown that%
\begin{equation}
\int_{\mathcal{C}_{\underline{v}}}d\underline{w}\prod_{j=1}^{2}\left(
\tilde{\phi}_{\nu}(v_{1}-w_{j})\tilde{\phi}_{\nu}(v_{2}-w_{j})\tilde{c}\left(
v_{2}-w_{j}\right)  \right)  \varphi\left(  w_{12},k\right)  =K(v_{12},k)
\label{KKK}%
\end{equation}
with%
\begin{align*}
\varphi\left(  w,k\right)   &  =\frac{\left(  1-w\right)  K(w,k+1)}%
{\tilde{\phi}_{\nu}(w)\tilde{\phi}_{\nu}(-w)\left(  w+1/\nu-k-1\right)  }\\
K(u,k)  &  =\frac{\Gamma\left(  1-\frac{1}{2}\nu-\frac{1}{2}\nu u\right)
\Gamma\left(  -\frac{1}{2}\nu+\frac{1}{2}\nu u\right)  }{\Gamma\left(
\frac{3}{2}-\frac{1}{2}\nu k-\frac{1}{2}\nu u\right)  \Gamma\left(  \frac
{1}{2}-\frac{1}{2}\nu k+\frac{1}{2}\nu u\right)  }.
\end{align*}
Note that for $k=1/\nu-2$%
\[
\Psi\left(  w\right)  =\frac{1}{\sin\frac{1}{2}\pi\nu\left(  w-1\right)
\sin\frac{1}{2}\pi\nu\left(  w+1\right)  }\varphi\left(  w,1/\nu-2\right)  .
\]
Similarly to (\ref{KKK}) we have here\footnote{This result was in addition
checked with Mathematica.}%
\begin{align}
Y(o)  &  =\int_{\mathcal{C}_{\underline{v}}}d\underline{w}\prod_{j=1}%
^{2}\left(  \tilde{\phi}_{\nu}(v_{1}-w_{j})\tilde{b}\left(  v_{1}%
-w_{j}\right)  \tilde{\phi}_{\nu}(v_{2}-w_{j})\tilde{c}\left(  v_{2}%
-w_{j}\right)  \right)  \Psi\left(  w_{12}\right) \nonumber\\
&  =\frac{K(o,k=1/\nu-2)}{\sin\frac{1}{2}\pi\nu\left(  o-1\right)  \sin
\frac{1}{2}\pi\nu\left(  o+1\right)  }=2\frac{K(o,k=1/\nu-2)}{\cos\pi\nu
-\cos\pi\nu o} \label{Ynu}%
\end{align}
with $o=o_{12}=u_{24}=v_{12}$. The arguments are as follows: The function
$Y(o)$ satisfies the equations (\ref{wat}) with the S-matrix eigenvalue
$\tilde{S}_{0}^{O(6)}$ of (\ref{SO6}). Therefore the minimal solution is
$Y^{\min}(o)=K(o,1/\nu-2)\sin\frac{1}{2}\pi\nu\left(  o-1\right)  \sin\frac
{1}{2}\pi\nu\left(  o+1\right)  $. Pinching at $w_{1}\rightarrow v_{1}%
-2/\nu,w_{2}\rightarrow v_{2}$ and produces a double pole at $o=1$, wich
implies (\ref{Ynu}).

Now we consider the p-function (\ref{p4s}) in (\ref{X}), then (up to a
constant)%
\[
Y_{p}(o)=K(o,k=1/\nu-2).
\]
This result is obtained by applying to the equations which correspond to
(\ref{x}) and (\ref{v3}) the formula%
\begin{multline*}
\frac{1}{2\pi i}\left(  \int_{\mathcal{C}_{a}}+\int_{\mathcal{C}_{b}}\right)
dz\Gamma(a-z)\Gamma(b-z)\Gamma\left(  c+z\right)  \Gamma\left(  d+z\right)
f(z)\\
=\Gamma\left(  1-c-d-a-b\right)  \Gamma\left(  c+a\right)  \Gamma\left(
d+a\right)  \Gamma\left(  c+b\right)  \Gamma\left(  d+b\right) \\
\times\left(  f(a)\frac{\sin\pi\left(  c+b\right)  \sin\pi\left(  d+b\right)
}{\pi\sin\pi\left(  a-b\right)  }-f(b)\frac{\sin\pi\left(  c+a\right)  \sin
\pi\left(  d+a\right)  }{\pi\sin\pi\left(  a-b\right)  }\right)
\end{multline*}
where $\mathcal{C}_{a}$ encloses the poles of $\Gamma(a-z)$ and $f(z+1)=f(z)$ holds.
\end{proof}

%\bibliographystyle{JHEP}
%\bibliography{../ref}

\begin{thebibliography}{10}

\bibitem{KN}
P.~P. Kulish and E.~R. Nissimov, \emph{{Conservation Laws in the Quantum
  Theory: cos phi in Two-Dimensions and in the Massive Thirring Model}},
  {\emph{JETP Lett.} {\bf 24} (1976) 220--223}. [Pisma Zh. Eksp. Teor.
  Fiz.24,247(1976)].

\bibitem{AKNP}
I.~{\relax Ya}. Arefeva, P.~P. Kulish, E.~R. Nissimov and S.~J. Pacheva,
  \emph{{Infinite Set of Conservation Laws of the Quantum Chiral Field in
  Two-Dimensional Space-Time}}, {\emph{LOMI-E-1-1978} (1977) }.

\bibitem{Ku}
P.~P. Kulish, \emph{{Factorization of the Classical and Quantum s Matrix and
  Conservation Laws}}, \href{http://dx.doi.org/10.1007/BF01079418}{\emph{Theor.
  Math. Phys.} {\bf 26} (1976) 132}. [Teor. Mat. Fiz.26,198(1976)].

\bibitem{KR}
P.~P. Kulish and N.~{\relax Yu}. Reshetikhin, \emph{{Diagonalization of GL(N)
  Invariant Transfer Matrices and Quantum N Wave System (Lee Model)}},
  \href{http://dx.doi.org/10.1088/0305-4470/16/16/001}{\emph{J. Phys.} {\bf
  A16} (1983) L591--L596}.

\bibitem{BKZ2}
H.~Babujian, M.~Karowski and A.~Zapletal, \emph{{Matrix Difference Equations
  and a Nested Bethe Ansatz}},
  \href{http://dx.doi.org/10.1088/0305-4470/30/18/019}{\emph{J. Phys.} {\bf
  A30} (1997) 6425--6450}, [\href{https://arxiv.org/abs/hep-th/9611006}{{\tt
  hep-th/9611006}}].

\bibitem{BFK1}
H.~M. Babujian, A.~Foerster and M.~Karowski, \emph{{The nested SU(N) off-shell
  Bethe ansatz and exact form factors}},
  \href{http://dx.doi.org/10.1088/1751-8113/41/27/275202}{\emph{J. Phys.} {\bf
  A41} (2008) 275202}.

\bibitem{BFK2}
H.~M. Babujian, A.~Foerster and M.~Karowski, \emph{{The form factor program: a
  review and new results- the nested SU(N) off-shell Bethe ansatz}},
  {\emph{SIGMA} {\bf 2} (2006) paper 082, 16 pages},
  [\href{https://arxiv.org/abs/hep-th/0609130}{{\tt hep-th/0609130}}].

\bibitem{BFK3}
H.~M. Babujian, A.~Foerster and M.~Karowski, \emph{{Exact form factors of the
  SU(N) Gross-Neveu model and 1/N expansion}},
  \href{http://dx.doi.org/10.1016/j.nuclphysb.2009.09.023}{\emph{Nucl. Phys.}
  {\bf B825} (2010) 396--425}, [\href{https://arxiv.org/abs/0907.0662}{{\tt
  0907.0662}}].

\bibitem{BFK6}
H.~M. Babujian, A.~Foerster and M.~Karowski, \emph{{The Nested Off-shell Bethe
  ansatz and O(N) Matrix Difference Equations}},
  \href{https://arxiv.org/abs/1204.3479}{{\tt 1204.3479}}.

\bibitem{BFK7}
H.~M. Babujian, A.~Foerster and M.~Karowski, \emph{{Exact form factors of the
  O(N) $\sigma$-model}},
  \href{http://dx.doi.org/10.1007/JHEP11(2013)089}{\emph{JHEP} {\bf 1311}
  (2013) 089}, [\href{https://arxiv.org/abs/1308.1459}{{\tt 1308.1459}}].

\bibitem{BFK8}
H.~M. Babujian, A.~Foerster and M.~Karowski, \emph{{Bethe Ansatz and exact form
  factors of the $O(N)$ Gross Neveu-model}},
  \href{http://dx.doi.org/10.1007/JHEP02(2016)042}{\emph{JHEP} {\bf 02} (2016)
  042}, [\href{https://arxiv.org/abs/1510.08784}{{\tt 1510.08784}}].

\bibitem{Mal}
O.~Aharony, S.~S. Gubser, J.~M. Maldacena, H.~Ooguri and Y.~Oz, \emph{{Large N
  field theories, string theory and gravity}},
  \href{http://dx.doi.org/10.1016/S0370-1573(99)00083-6}{\emph{Phys.Rept.} {\bf
  323} (2000) 183--386}, [\href{https://arxiv.org/abs/hep-th/9905111}{{\tt
  hep-th/9905111}}].

\bibitem{KZar}
V.~Kazakov and K.~Zarembo, \emph{{Classical / quantum integrability in
  non-compact sector of AdS/CFT}},
  \href{http://dx.doi.org/10.1088/1126-6708/2004/10/060}{\emph{JHEP} {\bf 0410}
  (2004) 060}, [\href{https://arxiv.org/abs/hep-th/0410105}{{\tt
  hep-th/0410105}}].

\bibitem{KT1}
M.~Karowski and H.~J. Thun, \emph{{Complete S matrix of the O(2N) Gross-Neveu
  model}}, \href{http://dx.doi.org/10.1016/0550-3213(81)90484-3}{\emph{Nucl.
  Phys.} {\bf B190} (1981) 61--92}.

\bibitem{BFK0}
H.~M. Babujian, A.~Foerster and M.~Karowski, \emph{{The nested SU(N) off-shell
  Bethe ansatz and exact form factors}},
  \href{https://arxiv.org/abs/hep-th/0611012}{{\tt hep-th/0611012}}.

\bibitem{BFK4}
H.~M. Babujian, A.~Foerster and M.~Karowski, \emph{{The form factor program: A
  review and new results, the nested SU(N) off-shell Bethe ansatz and the 1/N
  expansion}}, \href{http://dx.doi.org/10.1007/s11232-008-0042-7}{\emph{Theor.
  Math. Phys.} {\bf 155} (2008) 512--522}.

\bibitem{BKKW}
B.~Berg, M.~Karowski, P.~Weisz and V.~Kurak, \emph{{Factorized U(n) Symmetric s
  Matrices in Two-Dimensions}},
  \href{http://dx.doi.org/10.1016/0550-3213(78)90489-3}{\emph{Nucl. Phys.} {\bf
  B134} (1978) 125--132}.

\bibitem{BW}
B.~Berg and P.~Weisz, \emph{{Exact S Matrix of the Chiral Invariant SU($N$)
  Thirring Model}},
  \href{http://dx.doi.org/10.1016/0550-3213(78)90438-8}{\emph{Nucl. Phys.} {\bf
  B146} (1978) 205--214}.

\bibitem{KKS}
R.~Koberle, V.~Kurak and J.~A. Swieca, \emph{{Scattering Theory and 1/$N$
  Expansion in the Chiral {Gross-Neveu} Model}},
  \href{http://dx.doi.org/10.1103/PhysRevD.20.897,
  10.1103/PhysRevD.20.2638.2}{\emph{Phys. Rev.} {\bf D20} (1979) 897}.
  [Erratum: Phys. Rev.D20,2638(1979)].

\bibitem{KuS}
V.~Kurak and J.~A. Swieca, \emph{{Anti-particles as Bound States of Particles
  in the Factorized S Matrix Framework}},
  \href{http://dx.doi.org/10.1016/0370-2693(79)90758-5}{\emph{Phys. Lett.} {\bf
  B82} (1979) 289--291}.

\bibitem{ZZ4}
A.~B. Zamolodchikov and A.~B. Zamolodchikov, \emph{{Exact S Matrix of
  Cross-Neveu Elementary Fermions}},
  \href{http://dx.doi.org/10.1016/0370-2693(78)90738-4}{\emph{Phys. Lett.} {\bf
  B72} (1978) 481--483}.

\bibitem{K1}
M.~Karowski, \emph{{On the Bound State Problem in (1+1)-dimensional Field
  Theories}}, \href{http://dx.doi.org/10.1016/0550-3213(79)90600-X}{\emph{Nucl.
  Phys.} {\bf B153} (1979) 244--252}.

\bibitem{KW}
M.~Karowski and P.~Weisz, \emph{{Exact Form-Factors in (1+1)-Dimensional Field
  Theoretic Models with Soliton Behavior}},
  \href{http://dx.doi.org/10.1016/0550-3213(78)90362-0}{\emph{Nucl. Phys.} {\bf
  B139} (1978) 455--476}.

\bibitem{Sm}
F.~Smirnov, \emph{{Form Factors in Completely Integrable Models of Quantum
  Field Theory}}, {\emph{Adv. Series in Math. Phys. \textbf{14}, World
  Scientific} (1992) }.

\bibitem{BK}
H.~Babujian and M.~Karowski, \emph{{Exact form-factors in integrable quantum
  field theories: The sine-Gordon model. 2.}},
  \href{http://dx.doi.org/10.1016/S0550-3213(01)00551-X}{\emph{Nucl. Phys.}
  {\bf B620} (2002) 407--455},
  [\href{https://arxiv.org/abs/hep-th/0105178}{{\tt hep-th/0105178}}].

\bibitem{Wo}
\emph{\href{http://mathworld.wolfram.com/BarnesG-Function.html}{\tt
  http://mathworld.wolfram.com/barnesg-function}}, .

\bibitem{BFKZ}
H.~M. Babujian, A.~Fring, M.~Karowski and A.~Zapletal, \emph{{Exact
  form-factors in integrable quantum field theories: The Sine-Gordon model}},
  \href{http://dx.doi.org/10.1016/S0550-3213(98)00737-8}{\emph{Nucl. Phys.}
  {\bf B538} (1999) 535--586},
  [\href{https://arxiv.org/abs/hep-th/9805185}{{\tt hep-th/9805185}}].

\end{thebibliography}
\providecommand{\href}[2]{#2}\begingroup\raggedright\endgroup

\end{document}